\documentclass[reqno]{amsart}

\usepackage{amsmath}
\usepackage{amsfonts}
\usepackage{amssymb}
\usepackage{mathrsfs}
\usepackage[]{hyperref}

\newtheorem{theorem}{Theorem}
\theoremstyle{plain}

\newtheorem{lemma}{Lemma}
\newtheorem{proposition}{Proposition}
\newtheorem{remark}{Remark}

\numberwithin{equation}{section}

\newcommand{\reals}{\mathbb{R}}
\newcommand{\complexes}{\mathbb{C}}
\newcommand{\integers}{\mathbb{Z}}

\newcommand{\spec}[1]{\operatorname{spec}(#1)}
\newcommand{\proj}[1]{\mathsf{P}_{#1}}
\newcommand{\tr}[1]{\mathrm{tr} \, #1 }
\newcommand{\ptr}[2]{\mathrm{tr}_{#1} #2}
\newcommand{\id}[1]{\mathrm{id}_{#1}}
\newcommand{\comm}[2]{[#1,#2]}
\newcommand{\acomm}[2]{\{#1,#2\}}
\newcommand{\adj}{\prime}
\newcommand{\hadj}{*}
\renewcommand{\vec}[1]{\mathbf{#1}}

\newcommand{\hilbert}{\mathscr{H}}
\newcommand{\hilberte}{\hilbert_{\mathrm{e}}}
\newcommand{\matr}[1]{M_{#1}(\complexes)}
\newcommand{\matrd}{\matr{d}}

\newcommand{\hav}{\bar{H}}
\newcommand{\hame}{H_{\mathrm{e}}}
\newcommand{\hint}{H_{\mathrm{int.}}}
\newcommand{\envdens}{\rho_{\mathrm{e}}}

\begin{document}
\title[On the Lyapunov-Perron reducible...]{On the Lyapunov-Perron reducible Markovian Master Equation}
\author{Krzysztof Szczygielski}
\address[K. Szczygielski]{Institute of Theoretical Physics and Astrophysics, Faculty of Mathematics, Physics and Informatics, University of Gda\'{n}sk, Wita Stwosza 57, 80-308 Gda\'{n}sk, Poland}
\email[K. Szczygielski]{krzysztof.szczygielski@ug.edu.pl}

\begin{abstract}
We consider an open quantum system in $\matrd$ governed by quasiperiodic Hamiltonian with rationally independent frequencies and under assumption of Lyapunov-Perron reducibility of associated Schroedinger equation. We construct the Markovian Master Equation and resulting CP-divisible evolution in weak coupling limit regime, generalizing our previous results from periodic case. The analysis is conducted with application of projection operator techniques and concluded with some results regarding stability of solutions and existence of quasiperiodic global steady state.
\end{abstract}

\maketitle

\section{Introduction}

Completely positive evolution of open quantum systems governed by Hamiltonian depending on time acquires growing attention, which is partly due to increasing interest from quantum information community as well as a rapidly developing branch of quantum thermodynamics. In particular, evolution of open systems under \emph{periodic} Hamiltonian was successfully developed with the aid of \emph{Floquet theory} in \cite{Alicki2006b,Szczygielski2014} and led to some interesting results on both mathematics and physics grounds \cite{Szczygielski2013,Szczygielski2015,Gelbwaser-Klimovsky2015,Szczygielski2019,Szczygielski2020}. Recently, a case of system interacting with infinite environment by \emph{periodically modulated interaction Hamiltonian} was also considered in \cite{Szczygielski2020}. In this article, we generalize results of ref. \cite{Szczygielski2014} onto the case when the underlying Hamiltonian is not periodic, but rather \emph{quasiperiodic} (or multiperiodic), namely, when it can be expressed with a set of different, independent (non-commensurate) frequencies (see Section \ref{sec:LPreducibility} for precise definition). To ensure that this time dependence is not periodic, we impose additional assumption of \emph{rational independence} on a set of frequencies. A prototypical, physical realization of such scenario would be an atom modulated by not one, but few laser beams of different frequencies, which are mutually \emph{non-commensurate} (i.e.~are not integer multiples of the same basic frequency) and interacting weakly with surrounding electromagnetic field in thermal state.

This article is structured as follows. In Section \ref{sec:Background} we present some necessary mathematical preliminaries, including a notion of \emph{Lyapunov-Perron reducibility}. In Section \ref{sec:Model}, we draw a general model of open quantum system and characterize certain technical assumptions, under which the actual construction is performed. The main derivation of Markovian dynamics is then carried out in succeeding Section \ref{sec:CompletelyPositiveEvolution} in the formal regime of so-called \emph{weak coupling limit} and with application of projection operator technique. Our analysis is heavily inspired by results of Davies and Spohn \cite{Davies1974,Davies1976,Davies1978}, which laid the foundations for rigorous treatment of Markovianity in open systems with underlying Hamiltonian model. The main result of this section is formulated in Theorems \ref{thm:TheMainResult} and \ref{thm:LambdaCP}, where we formally construct completely positive dynamics in quasiperiodic setting, and in particular we demonstrate that under assumptions listed in Section \ref{sec:Assumptions}, the resulting Markovian Master Equation is Lyapunov-Perron reducible and a resulting quantum dynamical map admits a product structure. This corresponds with our previous results on periodic systems. Finally, in the succeeding Section \ref{sec:Stability} we formulate some results regarding asymptotic stability of solutions and existence of a global, quasiperiodic steady state.

\section{Reducible Markovian Master Equation}

\subsection{Background}
\label{sec:Background}

We use a standard notation. For normed space $(\mathscr{X},\|\cdot\|_\mathscr{X} )$, we denote by $\mathscr{B}(\mathscr{X})$ the Banach \emph{space of bounded linear maps} on $\mathscr{X}$ (with supremum norm). For Hilbert space $\hilbert$, symbol $\mathscr{B}_1 (\hilbert)$ will denote the Banach space of \emph{trace class operators} on $\hilbert$, complete with \emph{trace norm} $\| a \|_1 = \tr{\sqrt{a^\hadj a}}$ (${}^\hadj$ denotes Hermitian conjugation).

For matrix $m\in\matrd$, we will denote by $\| m \|$ the standard \emph{operator norm} of $m$, induced by $l^2$-norm $\| \cdot \|_2$ in $\complexes^d$. Duality pairing on $\matrd$ will be $(a,b) = \tr{ab}$ for $a,b\in\matrd$. Adjoint of linear map $T\in \mathscr{B}(\matrd)$ will be $T^\adj$ and subject to condition $\tr{[a\, T(b)]} = \tr{[T^\adj (a) \, b]}$ for all $a,b\in\matrd$.

Finally, we will be dealing with some Banach space-valued functions. For such a function $f : D \to \mathscr{X}$, $D\subseteq\reals$ we will define its \emph{supremum norm} by usual expression $\sup_{x\in\reals}\| f(x) \|_\mathscr{X}$. Occasionally, we will make use of other norms and if this is the case, we will define them whenever necessary.

\subsubsection{Ordinary differential equations}

We will be considering some linear ordinary differential equations (ODEs) in Banach spaces and so, we provide some necessary name conventions (see e.g.~\cite{Chicone2006}). Let $\mathscr{X}$ be a Banach space and set an non-autonomous initial value problem on $\mathscr{X}\times\reals_+$ (with $\reals_+ = [0,\infty)$),
\begin{equation}
	\dot{x}(t) = A_t (x(t)), \qquad x(0)=x_0,
\end{equation}
where $A_t \in \mathscr{B}(\mathscr{X})$, $t\in\reals_+$. Then, a differentiable operator-valued function $t\mapsto\Phi_t$, such that $x(t) = \Phi_t (x_0)$ is a solution for given initial condition $x_0 \in \mathscr{X}$, will be called the \emph{principal fundamental solution} of the problem in question. $\Phi_t$ satisfies an operator version of the ODE of a form $\dot{\Phi}_t = A_t \Phi_t$, and is always invertible, i.e.~$\Phi_{t}^{-1}$ exists for all $t\in\reals_+$. Then, a function $(t,s)\mapsto \Phi_{t,s} = \Phi_t\Phi_{s}^{-1}$, $t\geqslant s$, which translates solutions along integral curves, will be called the \emph{state transition matrix}, or \emph{propagator}, associated with solution $\Phi_t$ and will be subject to the \emph{Chapman-Kolmogorov identities} $\Phi_{t,t} = I$ and $\Phi_{t,t'}\Phi_{t',s} = \Phi_{t,s}$ for any $t' \in [s,t] \subset \reals_+$.

\subsubsection{Lyapunov-Perron reducibility}
\label{sec:LPreducibility}

Let $\mathbb{T}^r \simeq (\reals / 2\pi \integers)^r$ be the $r$--dimensional torus, $r\geqslant 1$. We set a \emph{vector of frequencies} $\vec{\Omega} = (\Omega_1, \, ... \, , \Omega_r )$, $\Omega_j > 0$, as well as a continuous torus winding $\vec{\vec{\theta}} : \reals_+ \to \mathbb{T}^r$,
\begin{equation}\label{eq:winding}
	\vec{\vec{\theta}}(t) = (\Omega_1 t, \Omega_2 t, \, ... \, , \Omega_r t).
\end{equation}
We will say that a matrix-valued function $t\mapsto A_t = [a_{ij}(t)] \in\matr{d}$, $t\in\reals_+$ is \emph{quasiperiodic} (or multiperiodic), if it can be expressed as a composition
\begin{equation}\label{eq:Acomposition}
	A_t = (\hat{A}\circ \vec{\theta})(t) = \hat{A}(\Omega_1 t, \Omega_2 t, \, ... , \, \Omega_r t),
\end{equation}
for some function $\hat{A} : \mathbb{T}^r \to \matr{d}$, $\hat{A} = [\hat{a}_{ij}]$, which satisfies periodicity condition $\hat{A}(\vec{\theta}) = \hat{A}(\vec{\theta} + 2\pi\vec{k})$ for any $\vec{\theta}\in\mathbb{T}^r$, $\vec{k}\in\integers^r$, i.e. is $2\pi$-periodic in each variable. With any such function $\hat{A}$, we associate its multidimensional \emph{Fourier series} 
\begin{equation}
	\hat{A}(\vec{\theta}) \sim \sum_{\vec{n}\in\integers^r} \hat{A}_{\vec{n}} e^{i\vec{n}\cdot\vec{\vec{\theta}}}, \quad \hat{A}_{\vec{n}} = \int\limits_{\mathbb{T}^r} \hat{A}(\vec{\theta}) e^{-i\vec{n}\cdot\vec{\theta}} dV,
\end{equation}
with $dV$ being a normalized Lebesgue measure restricted to cube $[0,2\pi]^r$ and $\vec{n}\cdot\vec{\theta}=n_1 \theta_1 + ... + n_r \theta_r$ is a usual dot product in $\reals^r$. Later on, we will conveniently limit ourselves only to functions admitting a Fourier series which converges uniformly (and thus pointwise everywhere) on $\mathbb{T}^r$ with respect to (any) matrix norm in $\matrd$.
\vskip\baselineskip
\noindent Now, consider an initial value problem on $\complexes^d \times \reals_+$,
\begin{equation}\label{eq:ODE}
	\dot{\vec{x}}(t) = A_t \vec{x}(t), \quad \vec{x}(0) = \vec{x}_0 ,
\end{equation}
where $t\mapsto A_t \in \matrd$ is quasiperiodic. We say that ODE given in \eqref{eq:ODE} is \emph{Lyapunov-Perron reducible} (reducible, for short) if and only if there exists a continuous, linear invertible change of variables $\vec{x} = P_t \vec{y}$, where $t \mapsto P_t$  is also quasiperiodic, such that \eqref{eq:ODE} reduces to a system
\begin{equation}\label{eq:ODE2}
	\dot{\vec{y}}(t) = X\vec{y}(t), \quad \vec{y}(0) = \vec{x}_0
\end{equation}
with constant matrix coefficient $X$. Then, solving \eqref{eq:ODE2} and applying inverse transformation, one obtains a principal fundamental solution of \eqref{eq:ODE} in factorized form
\begin{equation}\label{eq:solutionFloquetForm}
	\Phi_t = P_t e^{tX}.
\end{equation}
If $r=1$, i.e.~$A_t$ is periodic, \eqref{eq:ODE} is always reducible by virtue of Floquet theory \cite{Chicone2006} and \eqref{eq:solutionFloquetForm} is known as \emph{Floquet normal form} of a solution. In general case ($r > 1$) however, reducibility is not guaranteed and still remains an area of active research. One of most elegant results in this domain is due to Johnson and Sell \cite{Johnson1981} and correlates reducibility with smoothness of $\hat{A}$, non-resonance properties of $\vec{\Omega}$ and so-called \emph{Sacker-Sell spectrum} of associated irrational twist flow $(t,\phi)\mapsto\phi + \vec{\Omega}t$ on $\mathbb{T}^r$. In this work however, we do not explore conditions for reducibility of Markovian Master Equations as this is a problem reaching far beyond the scope of our analysis, but rather we construct appropriate completely positive quantum dynamics, assuming that the underlying Hamiltonian produces \emph{a priori} reducible Schroedinger equation.

\subsection{A model of open system}
\label{sec:Model}

We will be considering a finite-dimensional open quantum system S, described by Hilbert space $(\complexes^d ,\| \cdot \|_2 )$, which weakly interacts with infinite-dimensional environment E, of Hilbert space $\hilberte$, via a bounded interaction Hamiltonian
\begin{equation}\label{eq:Hint}
	\hint = \sum_{\mu} S_\mu \otimes R_\mu ,
\end{equation}
where $S_\mu \in \matrd$, $R_\mu \in \mathscr{B}(\hilberte)$; summation index $\mu$ runs over a finite set. System S is governed by quasiperiodic Hamiltonian $H_t$, while the Hamiltonian $\hame$ of E is constant. State of compound system $S+E$ is then expressed as time-dependent joint \emph{density operator} $\sigma_t$, $t\in\reals_+$, i.e.~a positive semi-definite, trace class operator of trace 1, contained in Banach space $\mathcal{B} = \mathscr{B}_1 (\complexes^d\otimes\hilberte)$. This state undergoes a \emph{reversible evolution} generated by a joint, time-dependent bounded Hamiltonian $H_{\text{se}}(t)$, characterized by
\begin{equation}
	H_{\text{se}}(t) = H_{\mathrm{f.}}(t) + \lambda H_{\text{int.}}, \quad H_{\mathrm{f.}}(t) = H_t\otimes I + I \otimes \hame ,
\end{equation}
where $H_{\mathrm{f.}}(t)$ may be understood as a ``free'', quasiperiodic Hamiltonian of S+E and $\lambda > 0$ is a small (compared to relevant energy scales) dimensionless parameter. For convenience, we restrict our analysis only to the case $t\mapsto H_{\mathrm{f.}}(t)$ being of class $\mathscr{C}^{\infty}(\reals_+)$, i.e.~with all matrix elements being smooth functions. This will allow for more freedom in representing certain functions in our model as pointwise convergent Fourier series (we hope that this restriction can be weakened without affecting validity of our results).
\vskip\baselineskip

\subsubsection{Assumptions of the construction}
\label{sec:Assumptions}

For technical reasons, in all our computations we will be utilizing the following three simplifying assumptions:

\begin{enumerate}
	\item \label{assm:RatInd} \emph{Rational independence of frequency vector $\vec{\Omega}$}. We assume that the frequencies set $\{\Omega_1 , \, ... \, , \, \Omega_r\}$ is \emph{rationally independent}, which means that the only vector $\vec{k}\in\integers^r$ satisfying equation $\vec{k}\cdot\vec{\Omega}=0$, is $\vec{k} = 0$. 
	\item \label{assm:Reducible} \emph{Reducibility of Shroedinger equation}. This is the main assumption of our analysis. Namely, we assume that the Schroedinger equation governed by $H_t$,
\begin{equation}
	\dot{\psi}(t) = -i H_t\psi(t), \quad \psi(t)\in\complexes^d,
\end{equation}
is reducible. In consequence, the associated unitary evolution operator $u_t$, subject to equation $\dot{u}_t = -i H_t u_t$, admits a product structure
\begin{equation}
	u_t = p_{t} e^{-i\hav t},
\end{equation}
for quasiperiodic unitary operator $p_{t}$ such that $p_0 = I$, and a Hermitian matrix $\hav$. Following earlier works, we will be calling $\hav$ the \emph{averaged Hamiltonian}. Eigenvalues $\epsilon_i$ of $\hav$ will be called \emph{quasienergies}, and spectrum of associated derivation $\comm{\hav}{\cdot\,}$, i.e.
	\begin{equation}
		\spec{\comm{\bar{H}}{\cdot\,}} = \{\omega = \epsilon - \epsilon' : \epsilon, \epsilon' \in \spec{\hav}\}
	\end{equation}
will be called the \emph{set of Bohr quasi-frequencies}.
	\item \label{assm:OmegaCF} \emph{$\vec{\Omega}$-congruence freedom}. Lastly, we assume that no pair $(\omega_1,\omega_2)$ of eigenvalues of derivation $\comm{\hav}{\cdot\,}$ exists such that $\omega_1 - \omega_2 = \vec{k}\cdot\vec{\Omega}$ for any vector $\vec{k}\in\integers^r \setminus \{\vec{0}\}$, i.e. the only permitted case when $\omega_1 - \omega_2 = \vec{k}\cdot\vec{\Omega}$ is such that $\vec{k}=\vec{0}$ and therefore $\omega_1 = \omega_2$.
\end{enumerate}

\begin{remark}
Prior to proceeding to the actual construction, we will make some comments on significance of the assumptions listed above.
\begin{enumerate}
	\item The rational independence of frequencies (assumption \ref{assm:RatInd}) assures that the winding $\theta$ defined in \eqref{eq:winding} is \emph{irrational} and the associated orbit $\{\vec{\theta} : t\in\reals_+\}$ is dense on $\mathbb{T}^r$. This in turn yields that the Hamiltonian $H_t$ is necessarily \emph{non-periodic}, i.e.~there exists no $T > 0$ such that $H_{t+T} = H_t$. This assures that the scenario described in this article is completely distinct from our previous results on periodic systems \cite{Alicki2006b,Szczygielski2013,Szczygielski2020,Szczygielski2014} (however the general form of solutions remains similar in structure).
	\item Reducibility of underlying Schroedinger equation (assumption \ref{assm:Reducible}) is crucial for validity of some key results, such as reducibility of Markovian Master Equation (stated in Theorem \ref{thm:LambdaCP}) and internal structure of certain semigroup generators under weak coupling limit regime. This assumption is however not necessary in the sense that even when abandoned, one is still able to develop some results on general form of the solution, since e.g.~weak coupling limit still remains well-defined, and the quasiperiodic case is then no different than a generic case of system under time-dependent Hamiltonian, as given e.g.~in \cite{Davies1978}. With this assumption however, we are able to give some deeper insight into actual structure of solutions. We also note that in presence of rational independence of frequencies, existence of product form of $u_t$ is mathematically quite a nontrivial demand, far from being granted, in contrast to a simply periodic case (or commensurate case, as is commonly assumed by some authors). 
	\item The $\vec{\Omega}$-congruence freedom assumption is stated purely for convenience, as it allows for significant simplification of certain series-like expressions. We note here, that in principle one should be able to obtain results similar to ours with this condition lifted.
\end{enumerate}

\end{remark}

\subsubsection{Restriction of dynamics to subspaces}

In the subsequent considerations, we will largely adapt techniques applied by e.g. Davies and Spohn \cite{Davies1974,Davies1978} in order to present a formal construction of Markovian Master Equations under weak coupling limit regime and with quasiperiodicity of underlying Hamiltonian. Specifically, we define a \emph{projection operator} $\proj{0}$ on full Banach space $\mathcal{B}$ via the partial trace operation,
\begin{equation}
	\proj{0}(a) = (\ptr{\hilberte}{a}) \otimes \envdens, \quad a\in\mathcal{B},
\end{equation}
for $\envdens$ being a constant state of the environment. Putting then $\proj{1}=\id{}-\proj{0}$, we split $\mathcal{B}$ into two subspaces $\mathcal{B}_{0}$ and $\mathcal{B}_{1}$,
\begin{equation}
	\mathcal{B} = \mathcal{B}_{0}\oplus\mathcal{B}_{1}, \quad \mathcal{B}_{0}=\proj{0}\mathcal{B}=\matrd \otimes \envdens, \quad \mathcal{B}_{1}=\proj{1}\mathcal{B}
\end{equation}
which describe reduced state of S and E, respectively. We introduce two following derivations on $\mathcal{B}$,
\begin{equation}\label{eq:Derivations}
	Z_t = -i \comm{H_{\mathrm{f.}}(t)}{\cdot\,}, \quad A = -i \comm{\hint}{\cdot\,},
\end{equation}
and define $A_{ij} = \proj{i}A\proj{j}$ for $i,j\in \{0,1\}$. Immediately, we see that $\mathcal{B}_0$ is an \emph{invariant subspace} of $Z_t$, since, by $\comm{\hame}{\envdens}=0$,
\begin{equation}
	Z_t (\rho\otimes\envdens) = -i \comm{H_t \otimes I + I \otimes \hame}{\rho\otimes\envdens} = -i \comm{H_t}{\rho} \otimes \envdens ,
\end{equation}
and so $Z_t$ is an endomorphism on $\mathcal{B}_0$. For convenience, we choose operators $R_\mu$ in interaction Hamiltonian \eqref{eq:Hint} to be of \emph{vanishing expectation values}, i.e.~$\tr{\envdens R_\mu} = 0$. This leads to condition $A_{00}=0$ and allows to rewrite derivation $A$ as a sum
\begin{equation}
	A=A_{11}+\Delta
\end{equation}
for $\Delta = A_{01}+A_{10}$. By general considerations \cite{Breuer2002,Alicki2006a,Rivas2012}, density operator $\sigma_t$ is then subject to the \emph{von Neumann equation}
\begin{equation}\label{eq:VonNeumann}
	\dot{\sigma}_t = (Z_t + \lambda A)(\sigma_t) = (Z_t + \lambda A_{11} + \lambda\Delta)(\sigma_t).
\end{equation}
\begin{proposition}\label{prop:Ut0Reducible}
Let the reducibility assumption \ref{assm:Reducible} apply and let the principal solution of Schroedinger equation be $u_t = p_{t} e^{-it\hav}$ with quasiperiodic unitary $p_{t}$ and Hermitian $\hav$. Then, equation
\begin{equation}\label{eq:ODEZt}
	\dot{x}_t = Z_t (x_t), \quad x_t \in \mathcal{B}
\end{equation}
is reducible and its principal fundamental solution $U_{t}$ admits a product form
\begin{equation}\label{eq:U0t}
	U_{t} = P_{t} e^{t\bar{Z}},
\end{equation}
where $P_{t},\bar{Z}$ are endomorphisms on $\mathcal{B}_0$ and function $t\mapsto P_{t}$ is quasiperiodic. Their action on $a\in\mathcal{B}$ is given by
\begin{align}\label{eq:PTetZdefinitions}
	P_{t}(a) = \left(p_{t}\otimes I\right) a \left(p_{t}^\hadj \otimes I\right), \quad \bar{Z}(a) = -i \comm{\hav\otimes I + I\otimes\hame}{a},
\end{align}
and, consequently,
\begin{equation}\label{eq:ExptZ}
	e^{t\bar{Z}} = e^{-it\comm{\hav}{\cdot\,}}\otimes e^{-it\comm{\hame}{\cdot\,}}.
\end{equation}
\end{proposition}

\begin{proof}
Showing all above claims demands only for differentiating proposed solution $U_{t}$ and confirming \eqref{eq:ODEZt}. With the aid of Schroedinger equation, one differentiates $u_t$ in order to find an equation obeyed by $p_{t}$ ,
\begin{equation}\label{eq:ptDerivative}
	\dot{p}_t = -i H_t p_{t}+ip_{t}\hav,
\end{equation}
which in turn allows to find, after some algebra, expression for $\dot{P}_{t}$,
\begin{equation}
	\dot{P}_{t} = -i \comm{H_t\otimes I}{\cdot\,}P_{t} + i P_{t}\comm{\hav\otimes I}{\cdot\,}.
\end{equation}
This can be then used to check, that
\begin{align}
	\dot{U}_t = \dot{P}_{t} e^{t\bar{Z}} + P_{t}\bar{Z}e^{t\bar{Z}} &= -i \comm{H_t\otimes I + I \otimes \hame}{\cdot\,} P_{t} e^{t\bar{Z}}\\
	&= -i \comm{H_{\mathrm{f.}}(t)}{\cdot\,} U_t,\nonumber
\end{align}
i.e.~equation \eqref{eq:ODEZt} is indeed satisfied with proposed $U_{t}$, as desired. Clearly, both maps $\bar{Z}$, $P_t$ are endomorphisms when restricted to $\mathcal{B}_0$. Validity of \eqref{eq:ExptZ} is then a simple consequence of form of $\bar{Z}$.
\end{proof}

\noindent The following two lemmas will be of particular importance later on (see Appendix \ref{app:Supplement} for proofs):

\begin{lemma}\label{lemma:PtConvergence}
If $t\mapsto H_{\mathrm{f.}}(t)$ is of class $\mathscr{C}^{\infty}(\reals_+)$, then function $t\mapsto P_t$ admits a Fourier series
\begin{equation}
	P_t = \sum_{\vec{n}\in\integers^r} \hat{P}_{\vec{n}} e^{i\vec{n}\cdot\vec{\Omega}t},
\end{equation}
converging uniformly (and thus pointwise) on $\reals_+$ with respect to operator norm in $\mathscr{B}(\matrd)$. 
\end{lemma}

\begin{lemma}\label{lemma:PexptZdecomposition}
Let $\{\omega\} = \spec{\comm{\bar{H}}{\cdot}}$, as before. Then, there exists a family of maps $\{Q_\omega\}\subset \mathscr{B}(\matrd)$, such that the following spectral decomposition applies,
\begin{equation}\label{eq:PexptZdecomposition}
	\proj{0}e^{t\bar{Z}} = \sum_{\omega} Q_\omega e^{-it\omega}, \quad t\in\reals_+ .
\end{equation}
\end{lemma}

The following general observation will provide a useful way of expressing solutions of perturbed differential equations in terms of unperturbed ones. We present this statement without proof (see a book by Kato \cite{Kato1966} or ref.~\cite{Szczygielski2020} for more details):

\begin{proposition}
Principal fundamental solution $U_{t}^{\lambda}$ of ODE
\begin{equation}\label{eq:ODEZtA}
	\dot{x}_t = (Z_t + \lambda A_{11})(x_t), \quad x_t \in\mathcal{B}
\end{equation}
can be expressed by a following integral formula
\begin{equation}\label{eq:Ulambda}
	U_{t}^{\lambda} = U_{t} + \lambda\int\limits_{0}^{t} U_{t,t'}A_{11}U_{t'}^{\lambda}dt'
\end{equation}
for $U_t$ being a solution of \eqref{eq:ODEZt}. Similarly, solution $V_{t}^{\lambda}$ of von Neumann equation \eqref{eq:VonNeumann} may be expressed in analogous fashion as
\begin{equation}\label{eq:Vdef}
	V_{t}^{\lambda} = U_{t}^{\lambda} + \lambda\int\limits_{0}^{t} U_{t,t'}^{\lambda} \Delta V_{t'}^{\lambda} dt'.
\end{equation}
\end{proposition}

We emphasize here that despite similarity between equations \eqref{eq:ODEZt} and \eqref{eq:ODEZtA}, reducibility of the latter can not be a priori assured and one is not allowed to demand a product structure of $U_{t}^{\lambda}$ similar to \eqref{eq:U0t}. Nevertheless, a formal perturbed expression \eqref{eq:Ulambda} remains valid.

\subsection{Completely positive evolution}
\label{sec:CompletelyPositiveEvolution}

In this section we will focus on deriving the dynamical equation, governing evolution of of subsystem S. We start with defining a linear epimorphism $W_{t}^{\lambda} : \mathcal{B} \to \mathcal{B}_{0}$ by
\begin{equation}\label{eq:Wdef}
	W_{t}^{\lambda} = \proj{0}V_{t}^{\lambda}\proj{0},
\end{equation}
which will effectively encode the evolution of S. This is achieved by defining \emph{reduced density matrix} $\rho_{t}^{\lambda}$ subject to formula
\begin{equation}
	\rho_{t}^{\lambda} \otimes \envdens = W_{t}^{\lambda}(\sigma_0),
\end{equation}
for some $\sigma_0 \in\mathcal{B}$. Introducing $\rho_0 = \ptr{\hilberte}{\sigma_0}$, we immediately have 
\begin{equation}
	\rho_{t}^{\lambda} = \ptr{\hilberte}{W_{t}^{\lambda}(\sigma_0)} = \Lambda_{t}^{\lambda}(\rho_0),
\end{equation}
where we have introduced a notion of one-parameter \emph{quantum dynamical map} $\Lambda_{t}^{\lambda}$. We will focus on structure of map $\Lambda_{t}^{\lambda}$ in so-called \emph{weak coupling limit}, which is mathematically realizable by performing a formal limiting procedure $\lambda \to 0^+$, which will be the main result of this section. In particular we will show that under weak coupling limit regime, there exists a well-defined function $t\mapsto \Lambda_{t}$ which satisfies a \emph{reducible} Markovian Master Equation (under previously introduced assumptions), and therefore admits appropriate factorized form, being, at the same time, a well-behaved quantum dynamics, i.e.~a completely positive, trace norm contraction on $\matrd$.

\subsubsection{Weak coupling limit}

By combining \eqref{eq:Vdef} with \eqref{eq:Wdef} and substituting recursively formula for $W^{\lambda}_{t}$ in place of $W_{t'}^{\lambda}$ under the integral, one obtains the celebrated \emph{integral Nakajima-Zwanzig equation}
\begin{equation}
		W^{\lambda}_{t} = \proj{0} U_{t} + \lambda^2 \int\limits_{0}^{t} dt' \int\limits_{0}^{t'} U_{t,t'} A_{01} U^{\lambda}_{t',t''} A_{10} W^{\lambda}_{t''} \, dt'' ,
\end{equation}
which can be shown from property $\proj{i}\proj{j} = \delta_{ij}\proj{j}$ and relation $\comm{\proj{0}}{U_t} = 0$, which can be checked by computation. Next, we apply two useful transformations: \emph{first}, we switch over to interaction picture generated by $Z_t$, i.e.~we define new map $\tilde{W}_{t}^{\lambda} = U_{t}^{-1} W_{t}^{\lambda}$, effectively ``rolling back'' the free part of evolution; \emph{second}, we introduce a rescaled time $t \to \lambda^2 t$ in order to arrive, after some algebra, at the equation
\begin{equation}\label{eq:Wtilde}
	\tilde{W}_{t}^{\lambda} = \proj{0} + \int\limits_{0}^{t} U_{\lambda^{-2}s}^{-1} K_{t-s,s}^{\lambda} U_{\lambda^{-2}s} \tilde{W}_{s}^{\lambda} \, ds,
\end{equation}
where the \emph{memory kernel} $K_{t,s}^{\lambda}$ is given by
\begin{equation}\label{eq:KlambdaKernel}
	K_{t,s}^{\lambda} = \int\limits_{0}^{\lambda^{-2}t} \left(U_{x+\lambda^{-2}s,\lambda^{-2}s}\right)^{-1} A_{01} U^{\lambda}_{x+\lambda^{-2}s,\lambda^{-2}s}A_{10} \, dx.
\end{equation}
Now we are ready to perform a formal weak coupling limit, i.e.~consider how the solution $\tilde{W}_{t}^{\lambda}$ behaves in regime of very small $\lambda$. Our general result, which we will formulate in next paragraphs, is heavily inspired by original works by Davies and Spohn \cite{Davies1974,Davies1978} and technically similar to \cite{Szczygielski2020}.

We start with a definition of so-called \emph{auto-correlation functions of the environment}, given as an expectation value
\begin{equation}\label{eq:AutocorrDef}
	f_{\mu\nu}(t) = \tr{\left(e^{i\hame t}R_{\mu}^{\hadj}e^{-i\hame t}R_\nu \envdens\right)}, \quad t\in\reals_+ .
\end{equation}
Just like in the original approach, integrability of those functions will be a necessary condition for consistency of our computations. In particular, the following simple lemma applies:

\begin{lemma}\label{lemma:L1}
Let all functions $f_{\mu\nu} \in L^1 (\reals_+)$ and let $\xi_y : \reals_+ \to \reals_+$ be given via
\begin{equation}
	\xi_y (t) = \left\| \left(U_{t+s,s}\right)^{-1} A_{01} U_{t+s,s}A_{10} \right\|
\end{equation}
with $\| \cdot \|$ denoting operator norm in $\mathscr{B}(\mathcal{B}_0)$. Then, $\xi_s \in L^1 (\reals_+)$ for all $s\in\reals_+$.
\end{lemma}

\begin{proof}
Computations are quite straightforward and in a large extent, identical to those in \cite{Szczygielski2020} so we are not too precise here. We put $U_{t,s} = U_{t} U_{s}^{-1}$ for $U_{t} = P_{t} e^{t\bar{Z}}$ by Proposition \ref{prop:Ut0Reducible}, such that for any $\rho\otimes\envdens\in\mathcal{B}_0$ we have
\begin{equation}
	\xi_{s}(t) = \sup_{\| \rho\otimes\envdens \|\leqslant 1} \left\|\proj{0} P_{s} e^{-t\bar{Z}} P_{t+s}^{-1} A P_{t+s} e^{t\bar{Z}} P_{s}^{-1} A (\rho\otimes\envdens)\right\|.
\end{equation}
Next, we apply definition \eqref{eq:Derivations} of map $A$, rewrite expression under the norm as appropriate double commutator and purposefully change $H_{\text{int.}}(t)$ to $H_{\text{int.}}(t)^{\hadj}$ in the outer commutator. After expanding and calculating the partial trace (as it appears in projection $\proj{0}$) we retrieve expressions for auto-correlation functions \eqref{eq:AutocorrDef} and arrive at an upper bound for $\xi_s (t)$ after simple algebra,
\begin{align}
	\xi_{s}(t) &\leqslant 4 \sum_{\mu\nu} \left|f_{\mu\nu}(t)\right| \| S_\mu \| \| S_\nu \|,
\end{align}
which, since $f_{\mu\nu} \in L^1 (\reals_+)$, is integrable for all $s$ such that
\begin{equation}
	\int\limits_{0}^{\infty} \xi_s (t) dx \leqslant C \max_{\mu,\nu}\| f_{\mu\nu} \|_1 
\end{equation}
for $\| \cdot \|_1$ being the $L^1$-norm and $C = 4 \sum_{\mu\nu}\| S_\mu \| \| S_\nu \| < \infty$.
\end{proof}

Let $t_* \in \reals_+$ and let $\mathscr{V} = \mathscr{C}_0 ([0,t_*], \mathcal{B})$ denote a Banach space of continuous, $\mathcal{B}$-valued functions on interval $[0,t_*]$, complete with respect to supremum norm
\begin{equation}
	\| \varphi \|_\mathscr{V} = \sup_{t\in [0,t_*]}\| \varphi(t)\|, \quad \varphi\in\mathscr{V}.
\end{equation}
On space $\mathscr{V}$, we define bounded linear Volterra operator $\mathcal{H}_\lambda$, acting on function $\varphi \in \mathscr{V}$ by
\begin{equation}\label{eq:HlambdaOperatorDef}
	\mathcal{H}_\lambda (\varphi)(t) = \int\limits_{0}^{t} U_{\lambda^{-2}s}^{-1} K_{t-s,s}^{\lambda} U_{\lambda^{-2}s}(\varphi(s)) \, ds,
\end{equation}
where kernel $K_{t,s}^{\lambda}$ is given by \eqref{eq:KlambdaKernel}. It is clear from a form of the integrand, that $\mathcal{H}_\lambda$ constitutes for a solution of the integral equation \eqref{eq:Wtilde}, in a sense that if we define a function $w_\lambda \in \mathscr{V}$ by setting $w_\lambda (t) = \tilde{W}^{\lambda}_{t}(\sigma_0)$ for some initial point $\sigma_0 \in \mathcal{B}$, then this function will be a solution to the equation
\begin{equation}
	w_\lambda = \proj{0}(\sigma_0) + \mathcal{H}_{\lambda} (w_\lambda).
\end{equation}
Thus, $\mathcal{H}_\lambda$ represents an evolution of reduced state of the system $S$, expressed in form of operator on functional space, rather than a map on $\matrd$. In what follows, we will be approximating this evolution by a simplified one, and validity of this approximation will be granted provided the coupling constant $\lambda$ is small enough. This approximation will be carried out in two stages just like in \cite{Davies1974}. For this to happen, we introduce, for $\zeta : \reals \to \mathscr{B}(\mathcal{B})$ being an operator-valued function, its \emph{local average} $\zeta^\sharp$ via formula
\begin{equation}\label{eq:TimeAvrg}
	\zeta^\sharp = \lim_{\delta\to\infty} \frac{1}{2\delta} \int\limits_{-\delta}^{\delta} U_{\lambda^{-2}t}^{-1} \zeta(t) U_{\lambda^{-2}t} \, dt .
\end{equation}

Next, we show that also in case of quasiperiodic setting, local averaging is a well-defined operation, which will lead eventually to correct formulation of dynamical equations.

\begin{proposition}
Let all $f_{\mu\nu} \in L^1 (\reals_+)$. Mapping $t\mapsto \tilde{K}_{t}^{\lambda}$, $t\in\reals_+$ given by integral expression
\begin{equation}\label{eq:KtildesDefinition}
	\tilde{K}_{t}^{\lambda} = \int\limits_{0}^{\infty} \left(U_{x+\lambda^{-2}t,\lambda^{-2}t}\right)^{-1} A_{01} U_{x+\lambda^{-2}t,\lambda^{-2}t}A_{10} \, dx
\end{equation}
is a well-defined, continuous function with values in $\mathscr{B}(\mathcal{B}_0)$ for all $\lambda \in\reals_+$. Moreover, under assumptions \ref{assm:RatInd}--\ref{assm:OmegaCF}, its local average $K = (\tilde{K}^{\lambda}_{\cdot})^{\sharp}$ is a bounded operator on $\mathcal{B}$, independent of $\lambda$.
\end{proposition}

\begin{proof}
We rewrite the integrand using product form \eqref{eq:U0t} of $U_t$, relation $\comm{\proj{0}}{U_t} = 0$ and definitions \eqref{eq:PTetZdefinitions} of $P_t$ and $e^{t\bar{Z}}$ in order to obtain
\begin{align}\label{eq:IntegrandIntegrability}
	\big( U_{x+\lambda^{-2}t,\lambda^{-2}t}\big)^{-1} &A_{01} U_{x+\lambda^{-2}t,\lambda^{-2}t}A_{10} \\
	&= \proj{0}e^{-x\bar{Z}}P_{x+\lambda^{-2}t}^{-1} A P_{x+\lambda^{-2}t} e^{x\bar{Z}} A \proj{0} \nonumber \\
	&= -\proj{0} \comm{e^{-x\bar{Z}}P_{x+\lambda^{-2}t}^{-1}(H_{\mathrm{int.}})}{\comm{H_{\mathrm{int.}}}{\proj{0}(\cdot)}}. \nonumber
\end{align}
After using \eqref{eq:Hint} and calculating the double commutator we see that the action of integrand over any $a\in\mathcal{B}$ is
\begin{equation}
	-\sum_{\mu\nu} \proj{0} \comm{S_{\mu}^{\lambda}(x+\lambda^{-2}t,t)\otimes R_{\mu}(x)}{\comm{S_{\nu}\otimes R_{\nu}}{\rho\otimes\envdens}}
\end{equation}
where $\rho\otimes\envdens = \proj{0}(a)$ and $S_{\mu}^{\lambda}(x,t)$, $R_{\mu}(x)$ are given by
\begin{equation}
	S_{\mu}^{\lambda} (x,t) = e^{i\bar{H}x} p_{x}^{\hadj} S_{\mu} p_{x} e^{-i\bar{H}x}, \quad R_{\mu} (x) = e^{i\hame x} R_\mu e^{-i\hame x}.
\end{equation}
Expanding the double commutator, we obtain
\begin{align}\label{eq:Integrand}
	-\sum_{\mu\nu}\proj{0} \comm{S_{\mu}^{\lambda}&(x+\lambda^{-2}t,t)^\hadj \otimes R_{\mu}(x)^{\hadj}}{S_\nu \rho \otimes R_\nu \envdens} \\
	&+ \sum_{\mu\nu} \proj{0} \comm{S_{\mu}^{\lambda}(x+\lambda^{-2}t,t)\otimes R_\mu (x)}{\rho S_{\nu}^{\hadj}\otimes\envdens R_{\nu}^{\hadj}},\nonumber
\end{align}
where we purposefully changed $H_{\mathrm{int.}}$ to $H_{\mathrm{int.}}^{\hadj}$ in order to end up with the usual form of auto-correlation functions. By Lemma \ref{lemma:PtConvergence}, matrix-valued functions $t\mapsto p_{t}^{\hadj} S_\mu p_{t}$ admit the Fourier series expansion
\begin{equation}\label{eq:SmuDecompFourier}
	p_{t}^{\hadj} S_\mu p_{t} = \sum_{\vec{n}\in\integers^r} \hat{S}_{\mu,\vec{n}} e^{i\vec{n}\cdot\vec{\Omega}t},
\end{equation}
converging pointwise everywhere in $\reals_+$. Now, let us introduce a following spectral decomposition of averaged Hamiltonian $\hav = \sum_{k=1}^{d} \epsilon_k P_k$ for $\epsilon_k \in \reals$ and $\{P_k\}$ being the spectral projections. This allows us to define, for each $\mu$, a set of new matrices $\{S_{\mu\vec{n}\omega}\}$ by setting
\begin{equation}\label{eq:SmunDEcompOmega}
	S_{\mu\vec{n}\omega} = \sum_{(k,l)\sim \omega} P_{k} \hat{S}_{\mu,\vec{n}} P_{l},
\end{equation}
where notation $(k,l)\sim\omega$ means that we perform the summation over such indices $k,l$, that $\epsilon_k - \epsilon_l = \omega$ (see also proof of Lemma \ref{lemma:PexptZdecomposition} in Appendix \ref{app:Supplement} for similar construction). One can then confirm with straight algebra that $S_{\mu\vec{n}\omega}$ satisfy relations
\begin{align}\label{eq:HScommRel}
	\comm{\hav}{S_{\mu\vec{n}\omega}} = \omega S_{\mu\vec{n}\omega}, \quad \comm{\hav}{S_{\mu\vec{n}\omega}^{\hadj}} = -\omega S_{\mu\vec{n}\omega}^{\hadj}, \\
	e^{i\hav t} S_{\mu\vec{n}\omega} e^{-i\hav t} = e^{i\omega t} S_{\mu\vec{n}\omega}.\nonumber
\end{align}
Notice, that \eqref{eq:SmuDecompFourier}, \eqref{eq:SmunDEcompOmega} and completeness of $\{P_k\}$ together yield
\begin{equation}\label{eq:SmunDecompositions}
	\hat{S}_{\mu,\vec{n}} = \sum_\omega S_{\mu\vec{n}\omega}, \quad S_\mu = \sum_{\vec{n}\in\integers^r} \hat{S}_{\mu,\vec{n}},
\end{equation}
with the latter series converging due to pointwise convergence of \eqref{eq:SmuDecompFourier}, and in consequence,
\begin{equation}\label{eq:SmxtDecomposition}
	S_{\mu}^{\lambda} (x,t) = \sum_{\vec{n}\in\integers^r}\sum_{\omega} e^{i\omega_{\vec{n}} x} S_{\mu\vec{n}\omega},
\end{equation}
with $\omega_\vec{n} = \omega + \vec{n}\cdot\vec{\Omega}$ being a \emph{shifted Bohr quasi-frequency}. Then, we calculate partial trace dictated by $\proj{0}$, and use \eqref{eq:SmxtDecomposition} to put \eqref{eq:Integrand} as
\begin{align}\label{eq:IntegrandExplicit}
	\sum_{\mu\nu}  \sum_{\vec{n},\vec{m}\in\integers^r}\sum_{\omega \omega'} &\Big[ e^{-i\omega_{\vec{n}} x} e^{-i(\omega_{\vec{n}}-\omega^{\prime}_{\vec{m}})\lambda^{-2}t}f_{\mu\nu}(x)\Big(S_{\nu\vec{m}\omega'} \rho S_{\mu\vec{n}\omega}^{\hadj}-S_{\mu\vec{n}\omega}^{\hadj} S_{\nu\vec{m}\omega'} \rho\Big) \\
	&+ e^{i\omega_{\vec{n}} x} e^{i(\omega_{\vec{n}}-\omega^{\prime}_{\vec{m}})\lambda^{-2}t}\overline{f_{\mu\nu}(x)} \Big( S_{\mu\vec{n}\omega}\rho S_{\nu\vec{m}\omega'}^{\hadj}-\rho S_{\nu\vec{m}\omega'}^{\hadj}  S_{\mu\vec{n}\omega}\Big)\Big] \nonumber
\end{align}
converging uniformly for all $t\in\reals_+$. Now it is clear that if $f_{\mu\nu} \in L^1 (\reals_+)$, then integral over $x\in [0,\infty)$ in \eqref{eq:KtildesDefinition} is well-defined, since after integrating \eqref{eq:IntegrandExplicit}, one obtains \emph{one-sided Fourier transforms} of auto-correlation functions, taken at shifted Bohr quasi-frequencies $\omega + \vec{n}\cdot\vec{\Omega}$ for $\vec{n}\in\integers^r$. By Lemma \ref{lemma:L1}, $\tilde{K}^{\lambda}_{t}$ is bounded on $\mathcal{B}$ for all $t,\lambda\in\reals_+$. Continuity is then obvious by continuous dependence of propagator $U_{t,s}$ on both $t,s$.

In next few paragraphs we will calculate the local average $K$ of function $t\mapsto \tilde{K}^{\lambda}_{t}$ and we will demonstrate e.g.~that it does not indeed depend on $\lambda$. We will, however achieve this in two distinct ways, putting $K$ into two equivalent forms. First of those will be convenient for proving certain convergence in succeeding Theorem \ref{thm:ConvergenceofNets}, and we will use the second one for showing complete positivity in Theorem \ref{thm:LambdaCP}. We apply formula \eqref{eq:TimeAvrg} directly to \eqref{eq:KtildesDefinition} to obtain
\begin{align}\label{eq:KtildeAverage}
	(\tilde{K}^{\lambda}_{\cdot})^{\sharp} = \lim_{\delta\to\infty} \frac{1}{2\delta} \int\limits_{-\delta}^{\delta} dt \int\limits_{0}^{\infty} \left(U_{x+\lambda^{-2}t}\right)^{-1} A_{01} U_{x+\lambda^{-2}t,\lambda^{-2}t}A_{10} U_{\lambda^{-2}t} \, dt,
\end{align}
where we extend $\tilde{K}_{t}^{\lambda}$ onto $t<0$ in straightforward manner. Next, we employ product form \eqref{eq:U0t} of $U_t$, relation $\comm{\proj{0}}{U_t} = 0$ and spectral decomposition \eqref{eq:PexptZdecomposition} in order to put the integrand in a form
\begin{align}\label{eq:IntegrandNonexplicit}
\left(U_{x+\lambda^{-2}t}\right)^{-1} &A_{01} U_{x+\lambda^{-2}t,\lambda^{-2}t}A_{10} U_{\lambda^{-2}t}  = \\
	&= \proj{0} e^{-(x+\lambda^{-2}t)\bar{Z}} A(x+\lambda^{-2}t) e^{x\bar{Z}} A(\lambda^{-2}t) \proj{0} e^{\lambda^{-2}t\bar{Z}} \nonumber \\
	& = \sum_{\omega,\omega'} e^{i\omega x} e^{i\lambda^{-2}(\omega-\omega')t} Q_\omega A(x+\lambda^{-2}t) e^{x\bar{Z}} A(\lambda^{-2}t) Q_{\omega'},\nonumber
\end{align}
where we introduced quasiperiodic function $t\mapsto A(t)$ given via
\begin{equation}\label{eq:AtDef}
	A(t) = P_{t}^{-1} A P_t .
\end{equation}
By Lemma \ref{lemma:PtConvergence}, both $P_t$, $P_{t}^{-1}$ admit uniformly convergent Fourier series and so, $A(t)$ admits such also,
\begin{equation}\label{eq:AtDefExp}
	A(t) = \sum_{\vec{n}\in\integers^r} \hat{R}_{\vec{n}} e^{i\vec{n}\cdot\vec{\Omega}t},
\end{equation}
where equality holds everywhere in $\reals_+$; this allows to cast the integrand into
\begin{equation}
	\sum_{\omega,\omega'} \sum_{\vec{n},\vec{m}\in\integers^r} e^{i\lambda^{-2}(\omega_\vec{n}-\omega^{\prime}_{\vec{m}})t} e^{i\omega_{\vec{n}}x}  Q_\omega \hat{R}_{\vec{n}} e^{x\bar{Z}} \hat{R}_{-\vec{m}} Q_{\omega'},
\end{equation}
converging pointwise in $\mathscr{B}(\mathcal{B})$. Carrying out the integration with respect to variable $t$, we see, that
\begin{equation}\label{eq:IntegralExpLambda}
	\lim_{\delta\to\infty} \frac{1}{2\delta} \int\limits_{-\delta}^{\delta} e^{i\lambda^{-2}(\omega_\vec{n}-\omega^{\prime}_{\vec{m}})t} dt = \delta_{\omega_\vec{n}\omega_{\vec{m}}^{\prime}},
\end{equation}
which may be checked with quick computation. The resulting Kronecker delta is nonzero if and only if both shifted quasi-frequencies agree, $\omega_{\vec{n}} = \omega^{\prime}_{\vec{m}}$, or if
\begin{equation}
	\omega - \omega' = (\vec{m}-\vec{n})\cdot\vec{\Omega},
\end{equation}
for some $\vec{n},\vec{m}\in\integers^r$. However, in presence of assumption \ref{assm:OmegaCF}, this is possible only in case $\omega = \omega'$. This yields, that $\delta_{\omega_\vec{n}\omega_{\vec{m}}^{\prime}}$ is nonzero only if
\begin{equation}
	(\vec{n}-\vec{m})\cdot\vec{\Omega} = 0,
\end{equation}
which, by rational independence of frequencies $\{\Omega_i\}$ is in turn possible only if $\vec{n}=\vec{m}$. This yields, that in fact $\delta_{\omega_\vec{n}\omega_{\vec{m}}^{\prime}} = \delta_{\omega\omega'}\delta_{\vec{n}\vec{m}}$ and \eqref{eq:KtildeAverage} transforms into the \emph{first form} of $K$,
\begin{equation}\label{eq:KlambdaLocalAver}
	K = (\tilde{K}^{\lambda}_{\cdot})^{\sharp} = \sum_{\omega}\sum_{\vec{n}\in\integers^r}\int\limits_{0}^{\infty} e^{i\omega_{\vec{n}}x}  Q_\omega \hat{R}_{\vec{n}} e^{x\bar{Z}} \hat{R}_{-\vec{n}} Q_{\omega} dx.
\end{equation}
To obtain the second form, we rewrite the integrand in \eqref{eq:KtildeAverage} in a fashion similar to \eqref{eq:IntegrandIntegrability} and \eqref{eq:Integrand}, obtaining
\begin{align}
	-\proj{0} \comm{&e^{-(x+\lambda^{-2}t)\bar{Z}} P^{-1}_{x+\lambda^{-2}t}(H_{\mathrm{int.}})}{\comm{e^{-\lambda^{-2}t\bar{Z}}P^{-1}_{\lambda^{-2}t}(H_{\mathrm{int.}})}{\rho\otimes\envdens}} = \\
	=-&\sum_{\mu,\nu}\proj{0} \comm{S_{\mu}^{\lambda}(x+\lambda^{-2}t)^{\hadj}\otimes R_{\mu}(x+\lambda^{-2}t)^{\hadj}}{S_{\nu}^{\lambda}(\lambda^{-2}t)\rho \otimes R_{\nu}(\lambda^{-2}t)\envdens}\nonumber\\
	+ &\sum_{\mu,\nu}\proj{0}\comm{S_{\mu}^{\lambda}(x+\lambda^{-2}t)\otimes R_{\mu}(x+\lambda^{-2}t)}{\rho S^{\lambda}_{\nu}(\lambda^{-2}t)^{\hadj} \otimes \envdens R_{\nu}(\lambda^{-2}t)^{\hadj}}\nonumber
\end{align}
within the same notation. Then, again using \eqref{eq:SmxtDecomposition} and calculating the partial trace, we obtain
\begin{align}
	\sum_{...} &e^{-i(\omega_{\vec{n}}-\omega^{\prime}_{\vec{m}})\lambda^{-2}t}\Big[ f_{\mu\nu}(x) e^{-i\omega_{\vec{n}}x} \left(S_{\nu\vec{m}\omega'}\rho S_{\mu\vec{n}\omega}^{\hadj} - S_{\mu\vec{n}\omega}^{\hadj}S_{\nu\vec{m}\omega'}\rho\right) \\
	&+ \overline{f_{\nu\mu}(x)} e^{i\omega^{\prime}_{\vec{m}}x} \left(S_{\nu\vec{m}\omega'}\rho S_{\mu\vec{n}\omega}^{\hadj}-\rho S_{\mu\vec{n}\omega}^{\hadj} S_{\nu\vec{m}\omega'}\right) \Big]\otimes\envdens \nonumber
\end{align}
where $\sum_{...}$ indicates summation over all indices $\mu,\nu$, $\omega,\omega'$ and $\vec{n},\vec{m} \in \integers^r$. Next, we integrate the above expression over $x\in[0,\infty )$, which is carried out by computing the one-sided Fourier transforms of functions $f_{\mu\nu}$ in such way, that
\begin{equation}
	\int\limits_{0}^{\infty} f_{\mu\nu}(x) e^{-i\omega x} dx = \frac{1}{2} h_{\mu\nu}(\omega ) + i\zeta_{\mu\nu}(\omega),
\end{equation}
where $h_{\mu\nu}(\omega)$ stands for a full (double-sided) Fourier transform of $f$, i.e.~$h_{\mu\nu}(x) = \int_{-\infty}^{\infty} f_{\mu\nu}(x) e^{-i\omega x} dx$, and $\zeta_{\mu\nu}(\omega)$ is often calculated with application of Sohotzki's formulas. By direct check, matrices $[h_{\mu\nu}(\omega)]_{\mu\nu}$ and $[\zeta_{\mu\nu}(\omega)]_{\mu\nu}$ are, for all $\omega$, Hermitian and $[h_{\mu\nu}(\omega)]_{\mu\nu}$ is positive semidefinite by Bochner's theorem \cite{Alicki2006a,Szczygielski2014}. Now it only remains to compute the integral over $t\in [-\delta,\delta]$ and then execute the limiting procedure $\delta\to\infty$, leading to the \emph{second form} of $K$,
\begin{equation}\label{eq:KlambdaLocalAver2}
	K(a) = \Big(-i \comm{\Delta H}{\rho} + \bar{D}(\rho)\Big) \otimes \envdens ,
\end{equation}
where $a\in\mathcal{B}$ is arbitrary, $\rho\otimes\envdens = \proj{0}(a)$, and $\Delta H$, $\bar{D}$ are given by
\begin{subequations}\label{eq:LambShiftD}
	\begin{equation}
		\Delta H = \sum_{\mu,\nu} \sum_{\vec{n}\in\integers^r}\sum_{\omega} \zeta_{\mu\nu}(\omega + \vec{n}\cdot\vec{\Omega}) S_{\mu\vec{n}\omega}^{\hadj} S_{\nu\vec{n}\omega},
	\end{equation}
	\begin{equation}
		\bar{D}(\rho) = \sum_{\mu,\nu}\sum_{\vec{n}\in\integers^r}\sum_{\omega} h_{\mu\nu}(\omega+\vec{n}\cdot\vec{\Omega}) \left( S_{\nu\vec{n}\omega}\rho S_{\mu\vec{n}\omega}^{\hadj} - \frac{1}{2}\acomm{S_{\mu\vec{n}\omega}^{\hadj}S_{\nu\vec{n}\omega}}{\rho}\right),
	\end{equation}
\end{subequations}
where $\acomm{a}{b} = ab+ba$ stands for the \emph{anticommutator} of matrices $a$ and $b$. Naturally, both forms of $K$ are independent of $\lambda$, as stated. 
\end{proof}

Using commutation relation \eqref{eq:HScommRel}, one checks after straightforward algebra, the following result:

\begin{proposition}\label{prop:CovProp}
We have $\comm{\Delta H}{\hav} = 0$ and $\comm{K}{\comm{\hav}{\cdot\,}}=0$ as maps on $\matrd$.
\end{proposition}

Second part of the above statement, namely commutativity of maps $K$ and $\comm{\hav}{\cdot\,}$, is technically analogous to the so-called \emph{covariance property} in time-independent case \cite{Alicki2006a,Alicki2006b,Szczygielski2013,Rivas2012}. We are now ready to formulate the crucial \emph{convergence theorem}, which validates the weak coupling limit procedure in our case. The subsequent formulation is heavily inspired by similar result of Davies (see Theorem 2.1 in \cite{Davies1974}), however we will modify original statements, as well as proof methods, in order to make it compatible with the quasiperiodic setting.

\begin{theorem}\label{thm:ConvergenceofNets}
Let us define two integral operators on $\mathscr{V}$ by their action on any function $\varphi \in \mathscr{V}$,
\begin{subequations}
	\begin{align}
		\tilde{\mathcal{H}}_\lambda (\varphi)(t) &= \int\limits_{0}^{t} U_{\lambda^{-2}s}^{-1} \tilde{K}_{s}^{\lambda} U_{\lambda^{-2}s}(\varphi(s)) \, ds, \\
		\mathcal{K}(\varphi)(t) &= \int\limits_{0}^{t} K(\varphi(s)) \, ds,\label{eq:Koperatordefinition}
	\end{align}
\end{subequations}
where $\tilde{K}_{t}^{\lambda}$ is defined by \eqref{eq:KtildesDefinition} and $K$ is its local average \eqref{eq:KlambdaLocalAver}. Then, if all auto-correlation functions $f_{\mu\nu}$ satisfy (Davies') integrability condition
\begin{equation}\label{eq:EpsilonIntegrability}
	\int\limits_{0}^{\infty} |f_{\mu\nu}(x)|(1+x)^\epsilon dx < \infty
\end{equation}
for some $\epsilon > 0$ and if assumptions \ref{assm:RatInd}--\ref{assm:OmegaCF} are met, then both nets of operators $(\mathcal{H}_\lambda)_\lambda$ and $(\tilde{\mathcal{H}}_\lambda)_\lambda$ converge to operator $\mathcal{K}$ in strong operator topology in $\mathscr{B}(\mathscr{V})$, as $\lambda \to 0^+$.
\end{theorem}

\begin{proof}
First, we show that $(\mathcal{H}_\lambda - \tilde{\mathcal{H}}_\lambda)_\lambda \to 0$ strongly in $\mathscr{B}(\mathscr{V})$. By recursive substitutions in \eqref{eq:Ulambda}, one can express propagator $U_{t,s}^{\lambda}$ as uniformly convergent series
\begin{equation}
	U_{t,s}^{\lambda} = U_{t,s} + \sum_{n=1}^{\infty} \lambda^n b_{n}^{\lambda}(t,s),
\end{equation}
where operators $b_{n}^{\lambda}(t,s)$ are defined by
\begin{equation}
	b_{n}^{\lambda}(t,s) = \int\limits_{s}^{t} dt_1 \int\limits_{s}^{t_1} dt_2 \, ... \int\limits_{s}^{t_{n-1}} U_t A_{11}(t_1)A_{11}(t_2) ... A_{11}(t_n)U_{s}^{-1} \, dt_n,
\end{equation}
and $A_{11}(t_k) = U_{t_k}^{-1}A_{11}U_{t_k}$ for each variable $t_k$, $k \in \{1,\, ... \, , \, n\}$. This allows to re-express $K_{t,s}^{\lambda}$ and find a rough estimation on $\| K_{t,s}^{\lambda}-\tilde{K}^{\lambda}_{s} \|$ as
\begin{align}\label{eq:KtsminusKs}
	\| K_{t,s}^{\lambda}-\tilde{K}^{\lambda}_{s}\| &\leqslant \int\limits_{\lambda^{-2}t}^{\infty} \left\| (U_{x+\lambda^{-2}s,\lambda^{-2}s})^{-1}A_{01}U_{x+\lambda^{-2}s,\lambda^{-2}s}A_{10} \right\| \, dx \\
	&+ \sum_{n=1}^{\infty} \lambda^n a_{n}^{\lambda}(t,s), \nonumber
\end{align}
where $a_{n}^{\lambda}(t,s) = \int_{0}^{\infty} \|(U_{x+\lambda^{-2}s,\lambda^{-2}s})^{-1} A_{01} b_{n}^{\lambda}(x+\lambda^{-2}s,\lambda^{-2}s) A_{10} \| dx$. The integral in \eqref{eq:KtsminusKs} vanishes as $\lambda\to 0^+$ by Lemma \ref{lemma:L1}. Moreover, one shows that the remaining series converges to 0 term by term by integrability condition \eqref{eq:EpsilonIntegrability} (by approach similar to \cite{Davies1974,Davies1978}) and so, $K_{t,s}^{\lambda}-\tilde{K}^{\lambda}_{s} \to 0$ uniformly for all $t,s$. This yields, for any $\varphi \in \mathscr{V}$,
\begin{equation}
	\|(\mathcal{H}_\lambda - \tilde{\mathcal{H}}_\lambda)(\varphi)\|_{\mathscr{V}}\leqslant \int\limits_{0}^{t_*} \| K_{t-s,s}^{\lambda} - \tilde{K}_{s}^{\lambda} \| \| \varphi(s) \| \, ds
\end{equation}
which then converges to 0 as $\lambda\to 0^+$ by dominated convergence theorem; thus, $(\mathcal{H}_\lambda - \tilde{\mathcal{H}}_\lambda)\to 0$ in strong operator topology over $\mathscr{V}$. The second part of a proof demonstrates that also $\tilde{\mathcal{H}}_\lambda \to \mathcal{K}$ strongly. Operator $\tilde{K}^{\lambda}_{t}$ may be put, after using \eqref{eq:U0t} and \eqref{eq:PexptZdecomposition} as
\begin{align}
	\tilde{K}_{t}^{\lambda} &= \sum_{\omega,\omega'} \int\limits_{0}^{\infty} e^{i\omega x} e^{i(\omega-\omega')\lambda^{-2}s} A(x+\lambda^{-2}s) e^{x\bar{Z}} A(\lambda^{-2}s) Q_{\omega'} \, dx \\
	&=\sum_{\omega,\omega'}\sum_{\vec{n},\vec{m}\in\integers^r} e^{i\lambda^{-2}(\omega_{\vec{n}}-\omega^{\prime}_{\vec{m}})s}e^{i\omega_{\vec{n}} x} Q_\omega \hat{R}_{\vec{n}} e^{x\bar{Z}} \hat{R}_{-\vec{m}} Q_{\omega'} \, dx, \nonumber
\end{align}
where functions $A(t)$ are as before given by $A(t) = P_{t}^{-1} A P_t$ and admit Fourier series expansions $A(t) = \sum_{\vec{n}\in\integers^r} \hat{R}_{\vec{n}}e^{i\vec{n}\cdot\vec{\Omega}t}$. Then, we use the \emph{first form} \eqref{eq:KlambdaLocalAver} of $K$ to write, for any $\varphi\in\mathscr{V}$,
\begin{align}
	\| (\tilde{\mathcal{H}}_{\lambda}& - \mathcal{K})(\varphi)(t)\| =\\
	&= \Bigg\| \sum_{\omega,\omega'}\sum_{\vec{n},\vec{m}\in\integers^r} \int\limits_{0}^{\infty} e^{i\omega_{\vec{n}}x} Q_\omega \hat{R}_{\vec{n}} e^{x\bar{Z}} \hat{R}_{-\vec{m}} Q_{\omega'} (\phi_\lambda (t)) \, dx \Bigg\|,\nonumber
\end{align}
where we introduced function $\phi_\lambda\in\mathscr{V}$ by setting
\begin{equation}\label{eq:Phit}
	\phi_\lambda(t) = \int\limits_{0}^{t} \left[e^{i\lambda^{-2}(\omega_{\vec{n}}-\omega^{\prime}_{\vec{m}})s} - \delta_{\omega\omega'}\delta_{\vec{n}\vec{m}}\right] \varphi(s) \, ds .
\end{equation}
Dominated convergence theorem combined with Riemann-Lebesgue lemma then shows $\|\phi_\lambda \|_\mathscr{V} \to 0$; this finally proves $\mathcal{H}_\lambda \to \mathcal{K}$ strongly and both nets $(\mathcal{H}_\lambda)$, $(\tilde{\mathcal{H}}_\lambda)$ converge to operator $\mathcal{K}$ in strong operator topology, as $\lambda\to 0^+$.
\end{proof}

\subsubsection{CP-divisible dynamics}

The above theorem allows us to obtain an approximated form of solution of the Nakajima-Zwanzig equation in question:

\begin{theorem}
\label{thm:TheMainResult}
Let $W_{t}^{\lambda} = U_{t}\tilde{W}_{t}^{\lambda}$ and define epimorphism $W_t : \mathcal{B}\to\mathcal{B}_0$, continuously depending on $t\in [0,t_*]$, by
\begin{equation}\label{eq:WtDefinition}
	W_{t} = \proj{0}U_{t} + \int\limits_{0}^{t} U_{t} K U_{s}^{-1} W_{s} \, ds,
\end{equation}
where $K$ is again the local average of function $t\mapsto \tilde{K}_{t}^{\lambda}$ and $U_t = P_t e^{t\bar{Z}}$. Let $v_0 \in\mathcal{B}$ be arbitrary and set functions $v,v_\lambda \in \mathscr{V}$
\begin{equation}
	v_\lambda (t) = W_{t}^{\lambda}(v_0), \quad v(t) = W_{t}(v_0), \quad t\in [0,t_*].
\end{equation}
Then, if assumptions \ref{assm:RatInd}--\ref{assm:OmegaCF} hold and if all auto-correlation functions $f_{\mu\nu}$ satisfy condition \eqref{eq:EpsilonIntegrability} for some $\epsilon > 0$, then $v_\lambda\to v$ uniformly over $[0,t_*]$ for any $v_0 \in \mathcal{B}$. Equivalently, $W_t$ represents a reduced evolution of system S in Schroedinger picture, i.e.~a solution of integral Nakajima-Zwanzig equation, suitable in regime of weak coupling limit $(\lambda\to 0^+)$.
\end{theorem}

\begin{proof}
We use methods similar to \cite[Theorem 2.1]{Davies1974} and \cite[Proposition 6]{Szczygielski2020}. Let $\mathcal{E}$ be an isometry on $\mathscr{V}$ given by equality $\mathcal{E}(\varphi)(t) = U_t (\varphi (t))$, representing a transformation from the interaction picture to the Schroedinger picture. Then, one checks that functions $v_\lambda$ and $v$ satisfy equations
\begin{equation}
	v_\lambda = a + \mathcal{E}\mathcal{H}_{\lambda}\mathcal{E}^{-1}(v_\lambda ), \quad v = a + \mathcal{E}\mathcal{K}\mathcal{E}^{-1}(v),
\end{equation}
for $a=\mathcal{E}\proj{0}(v_0)$, which comes from \eqref{eq:Wtilde}, \eqref{eq:HlambdaOperatorDef} and \eqref{eq:Koperatordefinition}. By recursive substitutions, this yields
\begin{equation}
	v_\lambda = \sum_{n=0}^{\infty} (\mathcal{E}\mathcal{H}_{\lambda}\mathcal{E}^{-1})^n (a) , \quad v = \sum_{n=0}^{\infty} (\mathcal{E}\mathcal{K}\mathcal{E}^{-1})^n (a).
\end{equation}
These in turn lead to the estimation
\begin{equation}
	\| v - v_\lambda \|_\mathscr{V} \leqslant \sum_{n=0}^{\infty} \| \mathcal{K}^{n}(a) - \mathcal{H}_{\lambda}^{n}(a) \|_\mathscr{V}
\end{equation}
which converges for all $a\in\mathscr{V}$, $\lambda\in\reals_+$ by the fact, that both $\mathcal{H}_\lambda$, $\mathcal{K}$ are Volterra integral operators and thus, norms of $\mathcal{K}^{n}$, $\mathcal{H}_{\lambda}^{n}$ are dominated by $(C t_{*})^{n}/n!$ for some constant $C>0$ (see \cite{Davies1974,Szczygielski2020} for more details). Each term of this series then converges to 0 as $\lambda\to 0^+$ by pointwise (strong) convergence of $\mathcal{H}_\lambda$ to $\mathcal{K}$ proved in Theorem \ref{thm:ConvergenceofNets}, and so, $v_\lambda\to v$ for all $v_0 \in \mathcal{B}$ uniformly in $\mathscr{V}$ over $[0,t_*]$ for every $t_* > 0$.
\end{proof}

We conclude this section by noting, that epimorphism $W_t$ defined in preceding Theorem \ref{thm:TheMainResult} gives rise to a \emph{legitimate} Markovian quantum dynamics on $\matrd$, i.e.~such a family $\{\Lambda_t : t\in\reals_+\}$ of maps on $\matrd$, which are \emph{trace preserving} and \emph{CP-divisible}. CP-divisibility of $\Lambda_t$ in this context means, that $\Lambda_t$ is infinitely divisible, i.e.~that for all $t\in\reals_+$ and all $s\in[0,t]$ there exists a propagator $\Lambda_{t,s}$ associated with the dynamics, such that $\Lambda_t = \Lambda_{t,s}\Lambda_s$, and that $\Lambda_{t,s}$ is a completely positive map on $\matrd$. In theory of open systems, CP-divisibility is generally accepted as a definition of Markovianity \cite{Chruscinski2014a}. Trace preservation condition is naturally equivalent to preservation of trace norm in a positive cone in $\matrd$.

\begin{theorem}\label{thm:LambdaCP}
Function $t \mapsto \Lambda_t$, $t\in\reals_+$ defined by equality
\begin{equation}
	W_t (v_0) = \Lambda_t (\rho_0)\otimes\envdens, \quad \rho_0 \otimes \envdens = \proj{0}(v_0),
\end{equation}
is uniformly continuous and differentiable, and map $\Lambda_t$ is CP-divisible and trace preserving on $\matrd$ for each $t\in\reals_+$, i.e.~a quantum dynamical map. Moreover, function $t\mapsto \rho_t = \Lambda_t (\rho_0)$ then satisfies Lyapunov-Perron reducible, Markovian Master Equation
\begin{equation}\label{eq:MMEODE}
	\dot{\rho}_t = L_t (\rho_t)
\end{equation}
for $L_t$ being a quasiperiodic Lindbladian in standard (Lindblad) form.
\end{theorem}

\begin{proof}
In this proof, we will use the \emph{second form} \eqref{eq:KlambdaLocalAver2} of $K$. Uniform continuity and differentiability of $W_t$ are clear from \eqref{eq:WtDefinition}. By Proposition \ref{prop:Ut0Reducible} and covariance property (Proposition \ref{prop:CovProp}), direct differentiation of $W_t$ gives
\begin{equation}
	\dot{W}_t = \left( Z_t + P_t K P_{t}^{-1}\right) W_t .
\end{equation}
Denote $\rho_t = \Lambda_t (\rho_0)$ such that $W_t (v_0) = \rho_t \otimes \envdens$ and introduce quasiperiodic function $t\mapsto \Sigma_t \in \mathscr{B}(\matrd)$, acting on $\rho\in\matrd$ by
\begin{equation}\label{eq:Sigma}
	\Sigma_t (\rho ) = p_t \, \rho \, p_{t}^{\hadj}, \quad \rho \in \matrd,
\end{equation}
being clearly an isometry, such that $P_t (\rho\otimes\envdens) = \Sigma_t (\rho) \otimes \envdens$ by local structure of $P_t$ \eqref{eq:PTetZdefinitions}. Immediately, from formula \eqref{eq:KlambdaLocalAver2} we see that $\mathcal{B}_0$ is an invariant subspace of $Z_t + P_t K P_{t}^{-1}$ and therefore, for any $v_0 \in\mathcal{B}$, we have
\begin{align}
	\dot{W}_t (v_0) &= (Z_t + P_t K P_{t}^{-1}) W_t (v_0) = (Z_t + P_t K P_{t}^{-1}) (\rho_t \otimes\envdens) \\
	&= \left( -i\comm{H_{t}}{\rho_t} + \Sigma_t (-i \comm{\Delta H}{\Sigma_{t}^{-1}(\rho_t)} + \bar{D}\Sigma_{t}^{-1}(\rho_t)) \right) \otimes \envdens \nonumber \\
	&= \left( -i\comm{H_t + \Sigma_t (\Delta H)}{\rho_t} + \Sigma_t \bar{D} \Sigma_{t}^{-1}(\rho_t)\right) \otimes\envdens . \nonumber
\end{align}
This means that the one-parameter map $\Lambda_t$ on $\matrd$ constitutes for a principal fundamental solution subject to operator ODE
\begin{equation}\label{eq:LambdaODE}
	\dot{\Lambda}_t = L_t \Lambda_t , \quad L_t = -i\comm{H_t + \Sigma_t(\Delta H)}{\cdot\,} + \Sigma_t \bar{D} \Sigma_{t}^{-1},
\end{equation}
where $\Sigma_t(\Delta H)$ is Hermitian and $\Sigma_t \bar{D} \Sigma_{t}^{-1}$ may be easily seen to be
\begin{align}
	\sum_{\mu,\nu}\sum_{\vec{n}\in\integers^r}\sum_{\omega} h_{\mu\nu}(\omega_\vec{n}) \left( S_{\nu\vec{n}\omega}^{t}\rho \left(S_{\mu\vec{n}\omega}^{t}\right)^{\hadj} - \frac{1}{2}\acomm{\left(S_{\mu\vec{n}\omega}^{t}\right)^{\hadj}S_{\nu\vec{n}\omega}^{t}}{\rho}\right),
\end{align}
for operators $S_{\nu\vec{n}\omega}^{t} = \Sigma_t (S_{\mu\vec{n}\omega})$. Now, by positive semi-definiteness of matrix $[h_{\mu\nu}(\omega_\vec{n})]_{\mu\nu}$, map $L_t$ is a bounded, quasiperiodic \emph{Lindbladian} in celebrated \emph{standard} (Lindblad-Gorini-Kossakowski-Sudarshan) \emph{form} \cite{Lindblad1976,Gorini1976} and therefore $\Lambda_t$ is a uniformly continuous, CP-divisible and trace preserving map on $\matrd$, i.e.~a \emph{legitimate} quantum dynamics in Schroedinger picture. Its explicit form is then given by expression
\begin{equation}\label{eq:LambdaExplicit}
	\Lambda_t = \Sigma_t e^{tX}, \quad X = -i\comm{\hav + \Delta H}{\cdot\,} + \bar{D},
\end{equation}
as we next check by direct computation. Derivative of map $\Sigma_t$ is, by \eqref{eq:ptDerivative},
\begin{equation}
	\dot{\Sigma}_t = -i \comm{H_t}{\cdot \,} \Sigma_t + i \Sigma_t \comm{\hav}{\cdot\,}
\end{equation}
which yields, with application of covariance property,
\begin{align}
	\dot{\Lambda}_t &= \dot{\Sigma}_t e^{tX} + \Sigma_t X e^{tX} \\
	&= \left( -i \comm{H_t}{\cdot \,} \Sigma_t + i \Sigma_t \comm{\hav}{\cdot\,} \right) e^{tX} + \Sigma_t \left( -i\comm{\hav + \Delta H}{\cdot\,} + \bar{D}\right) e^{tX} \nonumber \\
	&= \left( -i \comm{H_t + \Sigma_t (\Delta H)}{\cdot \,} + \Sigma_t \bar{D} \Sigma_{t}^{-1}\right) \Sigma_t e^{tX} \nonumber \\
	&= L_t \Lambda_t , \nonumber
\end{align}
so \eqref{eq:LambdaExplicit} satisfies ODE \eqref{eq:LambdaODE} and is therefore unique. The product form of $\Lambda_t$ then proves, that the Markovian Master Equation $\dot{\rho}_t = L_t (\rho_t)$, or equivalently its operator version \eqref{eq:LambdaODE}, is indeed reducible. This concludes the proof and the whole construction.
\end{proof}

We close this elaboration with a following remark. Notice, that map $X$ given by \eqref{eq:LambdaExplicit} is explicitly of standard (Lindblad) form, since $\bar{D}$ clearly is and so family $\{e^{tX} : t\in\reals_+\}$ is a completely positive $\mathcal{C}_0$-semigroup of trace norm contractions on $\matrd$, i.e.~a celebrated \emph{quantum dynamical semigroup}. Isometry $\Sigma_t$ is trivially completely positive and trace preserving, and so is map $\Lambda_t$ as a composition. 

\subsection{Stability of solutions and steady state}
\label{sec:Stability}
Here we make some notes on stability of the solutions, as well as the existence of steady state. In particular, we will show that stability may be deduced from spectral properties of map $X$ in \eqref{eq:LambdaExplicit}. This correlates closely with the case of \emph{periodic} Hamiltonian \cite{Szczygielski_2021}, i.e. if the underlying torus is of dimension $r=1$, and some of results proved there are applicable for this case as well.

Let $\xi \in \complexes$ be any eigenvalue of map $X$ given in \eqref{eq:LambdaExplicit} for eigenvector $\varphi\in\matrd$, so $X(\varphi) = \xi\varphi$. Also, let us define quasiperiodic function $\phi(t) = \Sigma_t (\varphi)$. Then,
\begin{equation}
	\varphi (t) = \Lambda_t (\varphi) = e^{t\xi} \phi(t)
\end{equation}
is a solution of ODE \eqref{eq:MMEODE}, i.e.~satisfies $\dot{\varphi}(t) = L_t (\varphi(t))$. Since $\Sigma_t$ is an isometry, we see that a qualitative, long-time behavior of $\varphi(t)$ is fully determined by $\xi$, or, in other words, by spectral properties of $X$. One quickly checks that three possible scenarios exist:
\begin{enumerate}
	\item If $\Re{\xi} < 0$, then the solution vanishes exponentially as $t\to\infty$.
	\item If $\Re{\xi} = 0$, then $\varphi (t) = e^{i (\Im{\xi}) t}\phi (t)$, i.e.~evolves isometrically.
	\item If finally $\Re{\xi} > 0$, then $\varphi (t)$ grows infinitely in norm, when $t\to\infty$.
\end{enumerate}

If $\Re{\mu} \leqslant 0$, the solution $\varphi (t)$ is called \emph{asymptotically stable}, and \emph{unstable} otherwise. This general classification originates in ODEs with \emph{periodic} matrix coefficient \cite{Yakubovich1975}, where numbers $\xi \in \spec{X}$ are called the \emph{characteristic exponents}. In periodic case, asymptotic behavior is equally determined by eigenvalues (\emph{characteristic multipliers}) of so-called \emph{monodromy matrix}, i.e.~a fundamental solution calculated at a \emph{period}. This is notably different situation, as no period exists in \emph{quasiperiodic} setting. Fortunately, the semigroup $\{e^{tX} : t\in\reals_+ \}$ was shown to be completely positive and trace preserving and thus, the unstable behavior of solutions will be shown to be impossible.

\begin{proposition}\label{prop:SpectraletX}
Let $\mathbb{D}^1$ be a unit disc in $\complexes$. Then, $\spec{e^{tX}}\subset\mathbb{D}^1$ for all $t\in\reals_+$, contains 1 and is invariant with respect to complex conjugation. In consequence, $\spec{X}$ is fully contained in complex left half-plane $\{\Re{\xi} \leqslant 0\}$, contains $0$ and is invariant with respect to complex conjugation.
\end{proposition}

\begin{proof}
Trace preservation condition demands that $(e^{tX})^\adj$, the Banach space adjoint of $e^{tX}$, is \emph{unital}, $(e^{tX})^\adj (I) = I$ for all $t\in\reals$. Thus, $I$ is one of its eigenvectors and necessarily $1 \in \spec{e^{tX}}$. Since $(e^{tX})^\adj$ is also completely positive, it attains its norm at $I$ \cite{Paulsen2003}, and so $\| (e^{tX})^\adj (I)\| = \| e^{tX} \|=1$, which then agrees with its spectral radius. This shows that $\spec{e^{tX}}$ is contained in $\mathbb{D}^1$. Map $e^{tX}$, being completely positive, is also Hermiticity preserving, which means that for all $a\in\matrd$, we have $e^{tX} (a)^{\hadj} = e^{tX}(a^\hadj)$. By this property, we see that if $\lambda_t$ is an eigenvalue for eigenvector $a$, then taking the Hermitian adjoint of eigenequation $e^{tX} (a) = \lambda_t a$ yields $a^\hadj$ is also an eigenvector for eigenvalue $\overline{\lambda_t}$. This shows that $\spec{e^{tX}}$ is either real, or consists of self-conjugate pairs of complex numbers. Remaining statements on $\spec{X}$ are naturally implied by properties of semigroup $e^{tX}$. Since $1\in\spec{e^{tX}}$ for all $t\in\reals_+$, then the spectral mapping theorem yields existence of $\xi_0 \in \spec{X}$ such that $e^{t\xi_0} = 1$ for every $t\geqslant 0$, which is possible only if $\xi_0 = 0$. Since $|e^{t\xi}| \leqslant 1$, we have $e^{t(\Re{\xi})} < 1$ and $\Re{\xi} \leqslant 0$. Finally, from complex conjugation invariance of $\spec{e^{tX}}$ we have that for every $\xi\in\spec{X}$, both $e^{t\xi},\overline{e^{t\xi}} \in \spec{e^{tX}}$, so $\xi,\overline{\xi}\in\spec{X}$ and spectrum of $X$ must be also invariant with respect to complex conjugation.
\end{proof}

Now, let us assume the map $X$ is \emph{diagonalizable}, i.e.~that it admits a set $\{\varphi_j\}$ of $d^2$ linearly independent eigenvectors, all subject to eigenequations $X(\varphi_j) = \xi_j \varphi_j$, $\xi_j \in \spec{X}$,  being a basis spanning $\matrd$. Let also
\begin{equation}
	\varphi_j (t) = e^{t\xi_j}\phi_j (t), \quad \phi_j (t) = \Sigma_t (\varphi_j).
\end{equation}
Then, each $\varphi_j (t)$ is automatically a solution and we can formulate a following result, being an immediate consequence of Proposition \ref{prop:SpectraletX}:

\begin{proposition}
All solutions $\varphi_j (t)$ of Markovian Master Equation \eqref{eq:MMEODE} for diagonalizable map $X$ are asymptotically stable.
\end{proposition}

For our next result, which is analogous to the one developed in \cite{Szczygielski_2021}, we introduce a following notation. For $\xi \in \spec{X}$, let us denote by $k_\xi$ its geometric multiplicity and by $E_{X}(\xi)$ a corresponding invariant eigenspace (so $k_\xi = \dim{E_{X}(\xi)}$).

\begin{proposition}

The following claims hold:
\begin{enumerate}
	\item If $\xi \in \spec{X}\setminus \{0\}$, then every $\varphi \in E_X (\xi)$ is traceless and thus cannot be positive semi-definite.
	\item If $\xi = 0$ is simple, i.e.~$k_0 = 1$, then $\varphi\geqslant 0$. If $k_0 > 1$, then there exists some positive semi-definite $\varphi^{+} \in E_X (0)$.
	\item If $\xi_j \in \spec{X}\setminus (-\infty , 0]$, then $\varphi \in E_X (\xi_j)$ and $\varphi^{\hadj} \in E_X (\overline{\xi_j})$. If in addition $\xi$ is simple, then $\varphi$ is Hermitian.
\end{enumerate}
\end{proposition}

Proof of this claim is virtually identical to the proof of Proposition 5 in \cite{Szczygielski_2021} and so we will not present it here. As a consequence, all above observations allow to formulate a following result on existence of asymptotic steady state, which is a direct analogue of periodic case (see \cite[Theorem 1]{Szczygielski_2021}).

\begin{proposition}\label{eq:etXspectral}
Let $\Lambda_t = \Sigma_t e^{tX}$ be characterized by \eqref{eq:Sigma} and \eqref{eq:LambdaExplicit} and let $X$ be diagonalizable. Then, $\Lambda_t$ admits a limit cycle $t\mapsto\rho_{t}^{\infty}$, such that each solution $\rho_t$ becomes arbitrarily close to $\rho_{t}^{\infty}$ in uniform topology in space $\mathscr{C}_0 ([t_0, \infty ), \matrd)$ of continuous, matrix-valued functions defined on $[t_0 ,\infty )$ for $t_0 > 0$ large enough. Moreover, $\rho_{t}^{\infty}$ is quasiperiodic if and only if no non-zero, purely imaginary eigenvalues of $X$ exist.
\end{proposition}

\begin{proof}
Proposition \ref{prop:SpectraletX} allows us to decompose $\spec{X}$ into three mutually disjoint subsets, $\spec{X} = \{0\} \cup M_1 \cup M_2$, where
\begin{equation}\label{eq:specXdecom}
	M_1 = \{\xi \in i \reals \setminus \{0\}\}, \quad M_2 = \{ \xi : \Re{\xi} < 0 \}.
\end{equation}
Let again $k_0 = \dim{E_X (0)}$ be the multiplicity of eigenvalue 0. By Proposition \ref{prop:SpectraletX}, eigenspace $E_{X}(0)$ is spanned by set of eigenvectors $\{\varphi_{0}^{(k)}\} \in \matrd$, $k \in \{1, \, ... \, , \, k_0\}$. Vectors $\varphi_{0}^{(k)}$ are semigroup invariants, $e^{tX}(\varphi_{0}^{(k)}) = \varphi_{0}^{(k)}$. Every solution $\rho_t$ can be, via decomposition \eqref{eq:specXdecom}, written as
\begin{align}
	\rho_t &= \sum_{k=1}^{k_0} c_{0}^{(k)} \phi_{0}^{(k)}(t) + \sum_{\xi_j \in M_1} c_{j} e^{i(\Im{\xi_j})t} \phi_{j}(t) \\
	&+ \sum_{\xi_j \in M_2} c_{j} e^{(\Re{\xi_j})t} e^{i(\Im{\xi_j})t} \phi_{j}(t), \nonumber
\end{align}
where $\phi_{0}^{(k)}(t) = \Sigma_t (\varphi_{0}^{(k)}(t))$, $\phi_{j}(t) = \Sigma_t (\varphi_{j}(t))$ and coefficients $c_{j}^{(k)}, c_{j} \in \complexes$ are determined by initial conditions. Condition $\Re{\xi_j} < 0$ for $\xi_j \in M_2$ then allows to deduce $\rho_{t}^{\infty}$ by deprecating the last sum,
\begin{equation}\label{eq:RhoInfinity}
	\rho_{t}^{\infty} = \sum_{k=1}^{k_0} c_{0}^{(k)} \phi_{0}^{(k)}(t) + \sum_{\xi_j \in M_1} c_{j} e^{i(\Im{\xi_j})t} \phi_{j}(t) .
\end{equation}
Indeed, we see that
\begin{align}
	\sup_{t\in[t_0,\infty )}\| \rho_t - \rho_{t}^{\infty} \|_1 &= \sup_{t\in [t_0, \infty )}\left\| \sum_{\xi_j \in M_2} c_{j} e^{(\Re{\xi_j})t} e^{i(\Im{\xi_j})t} \phi_{j}(t) \right\|_1 \\
	&\leqslant \sup_{t\in [t_0, \infty )} e^{-a t} \sum_{\xi_j \in M_2} |c_j| \| \varphi_j\|_1 \nonumber\\
	&= A e^{-at_0}, \nonumber
\end{align}
where $A = \sum_{\xi_j \in M_2} |c_j| \| \varphi_j\|$ and $a = \max|\Re{\xi_j}|$. Then, for any chosen $\epsilon > 0$, we have $\sup_{t\in[t_0,\infty )}\| \rho_t - \rho_{t}^{\infty} \|_1 < \epsilon$ if $t_0 > \frac{1}{a} \ln{\frac{A}{\epsilon}}$, i.e.~functions $\rho_t$, $\rho_{t}^{\infty}$ become arbitrarily close to each other in supremum norm. Quasiperiodicity condition will then be satisfied if and only if no $e^{i(\Im{\xi_j})t}$ factor appears under the second sum in \eqref{eq:RhoInfinity}, which is the case only if $\xi_j \in M_2$ are all real, or if $M_2 = \emptyset$. This is equivalent to the lack of imaginary $\xi_j$ other than 0. 
\end{proof}

\subsection{Note on the periodic case}

To conclude, we shortly remark on the most simple, however also most tractable case of the underlying torus being of dimension $r=1$, i.e.~a circle. In such scenario, Hamiltonian $H_t$ becomes simply \emph{periodic}. Naturally, periodicity remains a special case of quasiperiodicity, and so the majority of our analysis still applies. Moreover, in periodic setting the reducibility of initial Schroedinger equation does not need to be directly assumed, since it is actually \emph{granted} by virtue of \emph{Floquet theorem} and thus, assumptions \ref{assm:RatInd} and \ref{assm:Reducible} become redundant. The remaining assumption \ref{assm:OmegaCF} on $\Omega$-congruence freedom of set of Bohr quasi-frequencies however can still be invoked in order to end up with a mathematically pleasant, Lindblad form of semigroup generator $X$ in interaction picture (involving summation over only one copy of $\integers$). General results of this article remain valid in periodic case, i.e.~a product form of resulting dynamical map $\Lambda_t$ still applies. Similarly, the dynamical map also admits a steady state, periodic in this case. Periodic setting was already covered in details in literature, however results accessible therein were obtained with slightly different techniques; please see references \cite{Alicki2006b,Szczygielski2013,Szczygielski2014} for details.

\section{Conclusions}
\label{sec:Conclusions}

We demonstrated a derivation of reduced Markovian dynamics of open quantum system of finite dimension, governed by quasiperiodic Hamiltonian producing a Lyapunov-Perron reducible Schroedinger equation. In particular, we shown that the formal procedure of weak coupling limit is well-defined and applicable to our case, by employing a rigorous approach based on projection operator technique and integral Nakajima-Zwanzig equation. The derived Markovian Master Equation was shown to be also Lyapunov-Perron reducible in consequence. Its solution is a completely positive, trace preserving \emph{quantum dynamical map} of product form, clearly providing a straight generalization of our earlier results on periodically modulated systems. Both the mathematical structure of a solution, and its stability properties remain in close correlation with periodic case; the same can be then stated about validity of \emph{covariance property} and existence of time-dependent \emph{steady state} (\emph{quasiperiodic} in this case, as opposed to simply periodic one). We emphasized earlier, that one should be able to obtain comparable results (probably of much more involved formulation) if the assumption of $\vec{\Omega}$-congruence freedom of Bohr quasi-frequencies is lifted; the same can be then conjectured on abandoning the reducibility of Schroedinger equation. In such case however, no reducibility property of Master Equation can be hypothesized \emph{a priori}. This is a much tougher situation from mathematical point of view since one probably could not say much on internal structure of Lindbladian, apart from simply being of standard form (which would be granted by complete positivity and trace preservation of generated dynamics).

\section*{Acknowledgments}

This work was supported by the National Science Centre, Poland, via grant No. 2016/23/D/ST1/02043.

\appendix

\section{Mathematical supplement}
\label{app:Supplement}

\begin{proof}[Proof of Lemma \ref{lemma:PtConvergence}]
Notice, that if $H_{\mathrm{f.}}(t)$ depends smoothly on $t$, then via \eqref{eq:ptDerivative} there exist all derivatives of $p_t$ (since the $n$-th derivative demands for differentiating the Hamiltonian $n-1$ times) and so $t\mapsto p_t=[p_{jk}(t)]_{j,k=1}^{d}$ is also smooth. Since every function $p_{jk}$ may be put as a composition $p_{jk} = \hat{p}_{jk}\circ\vec{\theta}$ for $\hat{p}_{jk}$ defined on $\mathbb{T}^r$, we have $\hat{p}_{jk}\in\mathscr{C}^{\infty}(\mathbb{T}^r)$ and functions $\hat{p}_{jk}$ admit Fourier series
\begin{equation}\label{eq:pjkFourier}
	\hat{p}_{jk}(\theta) = \sum_{\vec{n}\in\integers^r} \hat{c}_{jk}(\vec{n}) e^{i\vec{n}\cdot\vec{\theta}}
\end{equation}
converging uniformly, and so pointwise, on $\mathbb{T}^r$ \cite{Grafakos2009}.

For any matrix $m\in\matrd$, let $\| m \|_{\mathrm{F}}$ be its \emph{Frobenius norm}, $\| m \|_{\mathrm{F}} = \sqrt{\tr{m^\hadj m}} = \left( \sum_{j,k=1}^{d}|a_{jk}|^{2}\right)^{1/2}$. One shows \cite{RogerA.Horn2012} that the inequality $\| m \| \leqslant \| m \|_{\mathrm{F}}$ holds for all $m\in\matrd$. This inequality allows as to estimate, for $N>0$,
\begin{equation}
	\sup_{\vec{\theta}\in\mathbb{T}^r}\Big\| \sum_{\|\vec{n}\|_2 <N}\hat{p}_{\vec{n}}e^{i\vec{n}\cdot\vec{\theta}} - \hat{p}(\vec{\theta}) \Big\| \leqslant \Bigg( \sum_{j,k=1}^{d}\sup_{\vec{\theta}\in\mathbb{T}^r}{\Big| \sum_{\|\vec{n}\|_{2}<N}\hat{c}_{jk}(\vec{n})e^{i\vec{n}\cdot\vec{\theta}}-\hat{p}_{jk}(\vec{\theta}) \Big|^{2}} \Bigg)^{\frac{1}{2}},
\end{equation}
where we introduced a notion of a \emph{circular partial sum} of multidimensional Fourier series \cite{Grafakos2009}. The right hand side of the inequality converges to 0 as $N\to\infty$ by uniform convergence of \eqref{eq:pjkFourier}, and so the Fourier series $\sum_{\vec{n}}\hat{p}_{\vec{n}}e^{i\vec{n}\cdot\vec{\theta}}$ converges uniformly to $\hat{p}$ with respect to matrix norm in $\matrd$. By equivalence of norms in $\matrd$, the result remains true for any matrix norm used.

Now, take $a\in\mathcal{B}$ and consider function $\vec{\theta}\mapsto \hat{P}_{\vec{\theta}}\in \mathscr{B}(\mathcal{B})$ such that $\hat{P}_{\vec{\theta}}(a) = (\hat{p}(\vec{\theta})\otimes I)\, m \, (\hat{p}(\vec{\theta})^{\hadj}\otimes I)$; this one will also admit uniformly convergent Fourier series. To see this, we set $N,M>0$ and estimate
\begin{equation}\label{eq:pProductEstimation}
	\sup_{\|a\|\leqslant 1}\Bigg\|\sum_{\|\vec{n}\|_{2}<N}\sum_{\|\vec{m}\|_{2}<M}(\hat{p}_{\vec{n}}\otimes I) \, a \, (\hat{p}_{\vec{m}}^{\hadj} \otimes I) e^{i(\vec{n}-\vec{m})\cdot\vec{\theta}} - \hat{P}_{\vec{\theta}}(a)\Bigg\|
\end{equation}
by adding and subtracting term $(\hat{p}(\vec{\theta})\otimes I) \, a \sum_{\|\vec{m}\|_{2}<M} (\hat{p}_{\vec{m}}^{\hadj}\otimes I) e^{-i\vec{m}\cdot\vec{\theta}}$ under the norm; after reshuffling and carrying out the supremum, \eqref{eq:pProductEstimation} is dominated by expression
\begin{align}
	&\leqslant \Big\|\sum_{\|\vec{n}\|_{2} <N} \hat{p}_{\vec{n}}e^{i\vec{n}\cdot\vec{\theta}} - \hat{p}(\vec{\theta})\Big\| \Big\| \sum_{\|\vec{m}\|_{2}<M}\hat{p}_{\vec{m}}e^{i\vec{m}\cdot\vec{\theta}}\Big\| + \Big\| \sum_{\|\vec{m}\|_{2}<M} \hat{p}_{\vec{m}}e^{i\vec{m}\cdot\vec{\theta}} - \hat{p}(\vec{\theta}) \Big\|
\end{align}
After employing uniform convergence of $\sum_{\vec{n}}\hat{p}_{\vec{n}}e^{i\vec{n}\cdot\vec{\theta}}$ and Moore-Smith theorem it is clear that Fourier series
\begin{equation}
	\sum_{\vec{n}\in\integers^r} \hat{C}_{\vec{n}} e^{i\vec{n}\cdot\vec{\theta}}, \quad \hat{C}_{\vec{n}} = \sum_{\vec{m}\in\integers^r} \hat{p}_{\vec{n}}\hat{p}_{\vec{n}-\vec{m}}^{\hadj} \otimes I
\end{equation}
of $\hat{P}_\vec{\theta}$ converges uniformly for all $\theta\in\mathbb{T}^r$ with respect to supremum norm in $\mathscr{B}(\mathcal{B})$ and the same is true for function $t\mapsto P_t$. This concludes the proof.
\end{proof}

\begin{proof}[Proof of Lemma \ref{lemma:PexptZdecomposition}]
Let the spectral decomposition of $\bar{H}$ be $\bar{H} = \sum_{k=1}^{d} \epsilon_k P_k$ (including multiplicities), where $\{P_k\}$ is a set of orthogonal projections. Define a set $\{q_\omega\}$ of maps on $\matrd$ given via equality, for any $\rho\in\matrd$,
\begin{equation}
	q_\omega (\rho) = \sum_{(k,l)\sim \omega} P_k \rho P_l,
\end{equation}
where notation $\sum_{(k,l)\sim \omega}$ indicates that the summation was performed only over such pairs $(k,l)$ of indices such that $\epsilon_k - \epsilon_l = \omega$. Applying completeness of $\{P_k\}$, one immediately writes
\begin{equation}
	\comm{\bar{H}}{\rho} = \sum_{k,l=1}^{d} (\epsilon_k - \epsilon_l) P_k \rho P_l = \sum_{\omega} \omega \, q_\omega (\rho).
\end{equation}
Then, one can check that \eqref{eq:ExptZ} along with spectral mapping theorem and relation $\comm{\hame}{\envdens} = 0$ yield, for any $a\in\mathcal{B}$,
\begin{align}
	\proj{0}e^{t\bar{Z}}\proj{0}(a) &= \sum_{\omega} e^{-i\omega t} \proj{0} \left( q_\omega (\rho) \otimes e^{-it\comm{\hame}{\cdot}}(\envdens) \right) \\
	&= \sum_\omega \sum_{\omega} e^{-i\omega t} \proj{0}(q_\omega \otimes \id{})\proj{0}(a),\nonumber
\end{align}
where $\rho = \proj{0}(a)$. Clearly, this is the claimed decomposition \eqref{eq:PexptZdecomposition} for spectral projections $Q_\omega = \proj{0}(q_\omega \otimes \id{})\proj{0}$.
\end{proof}


\begin{thebibliography}{10}

\bibitem{Alicki2006b}
R.~Alicki, D.~A. Lidar, and P.~Zanardi.
\newblock Internal consistency of fault-tolerant quantum error correction in
  light of rigorous derivations of the quantum {M}arkovian limit.
\newblock {\em Phys. Rev. A}, 73(5):052311, 2006.

\bibitem{Szczygielski2014}
K.~Szczygielski.
\newblock On the application of {F}loquet theorem in development of
  time-dependent {L}indbladians.
\newblock {\em J. Math. Phys.}, 55(8):083506, 2014.

\bibitem{Szczygielski2013}
K.~Szczygielski, D.~Gelbwaser-Klimovsky, and R.~Alicki.
\newblock Markovian master equation and thermodynamics of a two-level system in
  a strong laser field.
\newblock {\em Phys. Rev. E}, 87(012120):012120, 2013.

\bibitem{Szczygielski2015}
K.~Szczygielski and R.~Alicki.
\newblock Markovian theory of dynamical decoupling by periodic control.
\newblock {\em Phys. Rev. A}, 92(2):022349, 2015.

\bibitem{Gelbwaser-Klimovsky2015}
D.~Gelbwaser-Klimovsky, K.~Szczygielski, U.~Vogl, A.~Sa\ss{}, R.~Alicki,
  G.~Kurizki, and M.~Weitz.
\newblock Laser-induced cooling of broadband heat reservoirs.
\newblock {\em Phys. Rev. A}, 91:023431, 2015.

\bibitem{Szczygielski2019}
K.~Szczygielski and R.~Alicki.
\newblock {On Howland time-independent formulation of CP-divisible quantum
  evolutions}.
\newblock {\em Rev. Math. Phys.}, 32(07):2050021, 2020.

\bibitem{Szczygielski2020}
K.~Szczygielski.
\newblock Markovian dynamics under weak periodic coupling.
\newblock {\em J. Math. Phys.}, 62(1):012104, 2021.

\bibitem{Davies1974}
E.~B. Davies.
\newblock Markovian master equations.
\newblock {\em Commun. Math. Phys.}, 39(2):91--110, 1974.

\bibitem{Davies1976}
E.~B. Davies.
\newblock {\em {Q}uantum {T}heory of {O}pen {S}ystems}.
\newblock Academic Press, London, 1976.

\bibitem{Davies1978}
E.~B. Davies and H.~Spohn.
\newblock {Open quantum systems with time-dependent Hamiltonians and their
  linear response}.
\newblock {\em J. Stat. Phys.}, 19(5):511--523, 1978.

\bibitem{Chicone2006}
C.~Chicone.
\newblock {\em Ordinary {D}ifferential {E}quations with {A}pplications}.
\newblock Springer, New York, 2006.

\bibitem{Johnson1981}
R.~A. Johnson and G.~R. Sell.
\newblock Smoothness of spectral subbundles and reducibility of quasi-periodic
  linear differential systems.
\newblock {\em J. Differ. Equations}, 41(2):262--288, 1981.

\bibitem{Breuer2002}
H.-P. Breuer and F.~Petruccione.
\newblock {\em The theory of open quantum systems}.
\newblock Oxford University Press, New York, 2002.

\bibitem{Alicki2006a}
R.~Alicki and K.~Lendi.
\newblock {\em Quantum {D}ynamical {S}emigroups and {A}pplications}.
\newblock Springer, Berlin Heidelberg, 2006.

\bibitem{Rivas2012}
{\'{A}}.~Rivas and S.~F. Huelga.
\newblock {\em Open {Q}uantum {S}ystems: {A}n {I}ntroduction}.
\newblock Springer, Berlin Heidelberg, 2012.

\bibitem{Kato1966}
T.~Kato.
\newblock {\em Perturbation theory for linear operators}.
\newblock Springer Berlin Heidelberg, 1966.

\bibitem{Chruscinski2014a}
D.~Chru{\'{s}}ci{\'{n}}ski and S.~Maniscalco.
\newblock Degree of non-markovianity of quantum evolution.
\newblock {\em Phys. Rev. Lett.}, 112(12):120404, 2014.

\bibitem{Lindblad1976}
G.~Lindblad.
\newblock On the generators of quantum dynamical semigroups.
\newblock {\em Commun. Math. Phys.}, 48(2):119--130, 1976.

\bibitem{Gorini1976}
V.~Gorini, A.~Kossakowski, and E.~C.~G. Sudarshan.
\newblock Completely positive dynamical semigroups of {N}-level systems.
\newblock {\em J. Math. Phys.}, 17(5):821--825, 1976.

\bibitem{Szczygielski_2021}
K.~Szczygielski.
\newblock {On the Floquet analysis of commutative periodic Lindbladians in
  finite dimension}.
\newblock {\em Linear Algebra Appl.}, 609:176--202, 2021.

\bibitem{Yakubovich1975}
V.~A. Yakubovich and V.~M. Starzhinskii.
\newblock {\em Linear differential equations with periodic coefficients}.
\newblock John Wiley \& Sons, New York, 1975.

\bibitem{Paulsen2003}
V.~Paulsen.
\newblock {\em {Completely Bounded Maps and Operator Algebras}}.
\newblock Cambridge University Press, 2003.

\bibitem{Grafakos2009}
L.~Grafakos.
\newblock {\em Classical Fourier Analysis}.
\newblock Springer New York, 2009.

\bibitem{RogerA.Horn2012}
R.~A. Horn and C.~R. Johnson.
\newblock {\em Matrix Analysis}.
\newblock Cambridge University Press, 2012.

\end{thebibliography}
\end{document}